\documentclass[11pt]{amsart}

\usepackage{float,graphicx,equation,fullpage,amssymb,amsbsy,xspace}
\usepackage{adjustbox,blindtext}
\usepackage{psfrag}
\usepackage{xcolor}
\usepackage{amsmath}
\usepackage{epsfig}
	
\newcommand{\blot}[1]{}
\renewcommand{\d}{\mathrm{d}}
\newcounter{rem}

\newtheorem{theo}{Theorem}

\newtheorem{prop}[theo]{Proposition}

\newtheorem{remark}[rem]{Remark}

\newenvironment{Proof}{\removelastskip \vskip12pt plus 1pt \noindent
{\em Proof.\/}\rm }{\hfill$\Box$ \vskip12pt plus 1pt}

\title{Pendent Steady Rivulets: From Lubrication to Bifurcation}
\author{Michael Grinfeld and David Pritchard}

\address{Department of Mathematics and Statistics,
  University of Strathclyde, 26 Richmond Street, Glasgow G1 1XH, UK}

\email{david.pritchard@strath.ac.uk}

\begin{document}

\bibliographystyle{plain}

\begin{abstract}
We consider the shape of the free surface of steady pendent rivulets beneath a planar substrate. We formulate the governing equations in terms of two closely related dynamical systems and use classical phase-plane techniques to develop the bifurcation structure of the problem. Our results explain why lubrication theory is unable to capture this bifurcation structure for pendent rivulets, although it is successful in the related problem of sessile rivulets.
\end{abstract}

\maketitle

\section{Introduction\label{sec:intro}}

Rivulet flow on or beneath an inclined substrate has been the subject of many mathematical modelling studies, motivated by industrial and scientific applications, since the seminal work of Duffy \& Moffat \cite{DuffyMoffatt1995}. The physics of such flows is well understood, and implicit exact solutions to the full nonlinear problem have been presented (e.g. \cite{PerazzoGratton2004,Tanasijczuketal2010,Sokurov2020}). There is, however, a problem: in the case of a pendent, perfectly wetting rivulet, the classical lubrication approach (e.g. \cite{DuffyMoffatt1995,WilsonDuffy2005}), which is frequently used in parametric studies and extensions of the basic rivulet problem (e.g. \cite{Sullivanetal2008,Patersonetal2013,Patersonetal2014}), yields non-unique solutions which are defined only for particular rivulet widths.

The main purpose of the present study is to elucidate the relationship between the lubrication limit and the bifurcation structure of the full nonlinear boundary-value problem that describes the shape of a rivulet. This bifurcation structure has not, to the best of our knowledge, previously been investigated. The approach that we take also permits more analytical progress than previous treatments of the problem.

We confine ourselves to investigating the cross-sectional profile of uniform steady rivulets. The flux of fluid along the rivulet can therefore be ignored, and our results may also be interpreted as the profiles of static two-dimensional drops (fluid ridges) attached to a horizontal substrate. Two physically important quantities that characterise such a profile are its width and its total area; since solutions are symmetrical about the centreline, it will be convenient to consider the half-width and the half-area instead. 

We discuss the governing equations briefly, including the lubrication limit, in \S\ref{sec:goveqs}. In \S\ref{sec:phaseplanes} we formulate these equations in terms of two related phase planes, and describe the construction of solutions for perfectly wetting pendent rivulets in terms of orbits in these phase planes. In \S\ref{sec:bifurcation} we carry out a thorough bifurcation analysis of these solutions. In \S\ref{sec:imperfect} we discuss extensions to imperfectly wetting rivulets.

\section{Governing equations\label{sec:goveqs}}

The key physical assumption is the boundary condition at a liquid--air interface
\begin{equation}\label{eq:pgammaC}
 \hat P = \hat\gamma\hat C,
\end{equation}
where $\hat P$ is the excess pressure in the liquid relative to that in the air, $\hat\gamma$ is the coefficient of surface tension, and $\hat C$ is the (signed) curvature of the interface, defined such that if the liquid is on the same side as the centre of curvature then $\hat C>0$. (See e.g. \cite{Batchelor1967}, pp. 61--64.)

If we can write the free surface locally as $\hat z = \hat\zeta(\hat y)$ for some function $\hat\zeta(\hat y)$ then we may evaluate $\hat C$ as
\begin{equation}\label{eq:C}
 \hat C = \dfrac{\pm\hat\zeta''}{\left[1+(\hat\zeta')^2\right]^{3/2}},
\end{equation}
where the positive sign is taken if the liquid occupies $\hat z>\hat\zeta$, and the negative sign is taken if the liquid occupies $\hat z<\hat\zeta$.

\begin{remark}
This sign change is omitted from some presentations, leading to possible confusion. For example, the right-hand part of figure 3 in \cite{Tanasijczuketal2010} is not physically possible in the orientation shown.
\end{remark}

We will be concerned with points $\hat y = \hat y_0$  on the free surface such that $|\hat\zeta'|\to\infty$ as $\hat y\to\hat y_0$. At such points the pressure, and thus from \eqref{eq:pgammaC} the signed curvature $\hat C$, must be continuous. Below, we will call such points `critical points'.

Suppose that for $\hat y<\hat y_0$ there are two interfaces $\hat z=\hat\zeta_1(\hat y)$ and $\hat z=\hat\zeta_2(\hat y)>\hat\zeta_1(\hat y)$, with liquid occupying $\hat z<\hat\zeta_1$ and $\hat z>\hat\zeta_2$, and suppose that these interfaces meet at $\hat y=\hat y_0$ (so $\hat y_0$ marks the point at which the free surface curves back and starts to overhang itself). From the discussion above, the appropriate conditions at the critical point $\hat y_0$ are thus
\begin{equation}\label{eq:critBCs}
\begin{cases}
 \hat\zeta_1(\hat y_0) = \hat \zeta_2(\hat y_0),\\
 \displaystyle\lim_{\hat y\to\hat y_0} \dfrac{-\hat\zeta_1''}{\left[1+(\hat\zeta_1')^2\right]^{3/2}} = \displaystyle\lim_{\hat y\to\hat y_0} \dfrac{+\hat\zeta_2''}{\left[1+(\hat\zeta_2')^2\right]^{3/2}},
\end{cases}
\end{equation}
where the appropriate one-sided limit is assumed. It is simple to see that conditions \eqref{eq:critBCs} also apply if the liquid occupies $\hat\zeta_1<\hat z<\hat\zeta_2$, or if the interfaces exist for $\hat y>\hat y_0$.

We consider uniform steady pendent rivulets below a planar substrate which is inclined at an angle $\alpha$ to the horizontal; $\hat z$ is the coordinate perpendicular to the plane and increasing \emph{downward}, and $\hat y$ is a horizontal coordinate parallel to the plane. Within such a rivulet the pressure is hydrostatic and
\begin{equation}\label{eq:hydrostatic}
 \hat P(\hat y,\hat z) = \hat\gamma\hat C + \hat\rho\hat g\cos(\alpha)(\hat z-\hat\zeta(\hat y)),
\end{equation}
where $\hat\rho$ is the density of the liquid and $\hat g$ is the acceleration due to gravity, and where we evaluate $\hat C$ on the interface $\hat z=\hat\zeta$. (Note that if the point $(\hat y,\hat z)$ in the fluid lies between two interfaces, $\hat\zeta_1<\hat z<\hat\zeta_2$, then we may take either $\hat\zeta=\hat\zeta_1$ or $\hat\zeta=\hat\zeta_2$ in \eqref{eq:hydrostatic}.)

The hydrostatic condition $\partial\hat P/\partial\hat y = 0$ therefore becomes
\begin{equation}
 \hat\gamma\dfrac{\d\hat C}{\d\hat y} - \hat\rho\hat g\cos(\alpha)\dfrac{\d\hat\zeta}{\d\hat y} = 0,
\end{equation}
or on any given portion of the free surface,
\begin{equation}\label{eq:zetaode-dimensional}
 \pm\hat k\hat\zeta' = \left[\dfrac{\hat\zeta''}{\left[1+(\hat\zeta')^2\right]^{3/2}} \right]',
\end{equation}
where the positive sign is taken for sections of the free surface on which liquid lies below air, and the negative sign for sections on which liquid lies above air, and where the parameter $\hat k$ is given by
\begin{equation}
 \hat k = \frac{\hat\rho\hat g}{\hat \gamma} \cos(\alpha) > 0.
\end{equation}

Equation \eqref{eq:zetaode-dimensional} is to be solved subject to suitable boundary conditions at the contact lines $\hat y=\pm\hat a$,
\begin{equation}
 \hat\zeta(\pm\hat a) = 0 \qquad \text{and} \qquad \hat\zeta'(\pm\hat a) = \mp\tan(\beta),
\end{equation}
where $\beta\geq 0$ is the contact angle and $\hat a$ is the half-width of the rivulet.

We may scale out $\hat k$ by defining
\begin{equation}
 y = \sqrt{\hat k}\hat y, \qquad \zeta(y) = \sqrt{\hat k}\hat\zeta(\hat y), \qquad a = \sqrt{\hat k}\hat a,
\end{equation}
to obtain the dimensionless system
\begin{equation}\label{eq:zetaode}
 k\zeta' = \left[\dfrac{\zeta''}{[1+(\zeta')^2]^{3/2}}\right]'
\end{equation}
subject to
\begin{equation}\label{eq:zetabc}
 \zeta(\pm a) = 0 \qquad \text{and} \qquad \zeta'(\pm a) = \mp\tan(\beta).
\end{equation}
The parameter $k$ takes only two values: $k=+1$ corresponds to a section of the interface with liquid below air, while $k=-1$ corresponds to a section of the interface with liquid above air.

Unless otherwise stated, we consider the case of a perfectly wetting fluid, for which the contact angle $\beta=0$ and \eqref{eq:zetabc} becomes
\begin{equation}\label{eq:zetabc2}
 \zeta(\pm a) = 0 \qquad \text{and} \qquad \zeta'(\pm a) = 0.
\end{equation}

\subsection{Lubrication limit\label{sec:lubrication}}

The lubrication limit considered by Duffy and co-workers (e.g. \cite{DuffyMoffatt1995,Sullivanetal2008,Patersonetal2013,Patersonetal2014,WilsonDuffy2005}) corresponds to $|\zeta'|\ll1$. The system \eqref{eq:zetaode} and \eqref{eq:zetabc2} thus becomes
\begin{equation}\label{eq:zetaode_lub}
 k \zeta'= \zeta''', \qquad \zeta(\pm a) = 0, \qquad \zeta'(\pm a) = 0.
\end{equation}
When $k=-1$ (a thin pendent rivulet), \eqref{eq:zetaode_lub} admits the family of solutions \cite{WilsonDuffy2005}
\begin{equation}\label{eq:zetasol_lub}
 \zeta(y) = \dfrac{1}{2}h_{\mathrm{m}}\left[(1 - (-1)^n\cos(y)\right]
\end{equation}
defined only for $a=n\pi$ ($n\in\mathbb{N}$) but for any value of $h_{\mathrm{m}}>0$. Such a solution represents $n$ identical adjacent rivulets, each of maximum thickness $h_{\mathrm{m}}$; the total half-area $A$ of the solution is given by $A = \frac{1}{2}h_{\mathrm{m}}n\pi$. Thus in the lubrication limit, for a given solution branch indexed by $n$ the half-width $a$ is independent of $A$ but the maximum thickness $h_{\mathrm{m}}$ may be chosen freely to satisfy a given choice of $A$. As we will see, these solutions fail to capture the behaviour of the system for rivulets of non-vanishing thickness.

When $0<\beta<\pi/2$, the lubrication solution satisfying \eqref{eq:zetabc} becomes
\begin{equation}\label{eq:zetasol_lub_pb}
 \zeta(y) = \dfrac{\tan(\beta)}{\sin(a)}\left(\cos(y)-\cos(a)\right),
\end{equation}
with a single free parameter $a$. We shall return to this in \S\ref{sec:imperfect}.

\begin{remark}
When $k=1$ (a thin sessile rivulet), \eqref{eq:zetaode_lub} has no non-zero solution. The absence of a solution in the lubrication limit does reflect the behaviour of solutions to the full (non-lubrication) system; we discuss this briefly in \S\ref{sec:discussion}.
\end{remark}

\section{Phase planes and construction of solutions\label{sec:phaseplanes}}

\subsection{Phase planes\label{subsec:phaseplanes}}

We define
\begin{equation}
 p(y) = \zeta'(y) \qquad \text{and} \qquad q(y) = \dfrac{p'}{k(1+p^2)^{3/2}},
\end{equation}
so that \eqref{eq:zetaode} becomes the system
\begin{equation}\label{eq:pqodes}
  \begin{cases}
    p'= k(1+p^2)^{3/2}q,\\
    q'= p.
  \end{cases}
\end{equation}
It will be helpful when interpreting the solutions to note that while $p(y)$ describes the local gradient of the free surface, $q(y)$ represents both the pressure within the rivulet and (to within an additive constant) the profile $\zeta(y)$.

The contact-line conditions (\ref{eq:zetabc2}b) become
\begin{equation}\label{eq:pbc}
 p(\pm a) =0.
\end{equation}
We will see later how to deal with the conditions $\zeta(\pm a)=0$ [equation (\ref{eq:zetabc2}a)].

At a critical point $y=y_0$, two sections of the free surface $z=\zeta_1$ and $z=\zeta_2$ meet; the values of $k$ on these two sections must be different. In the obvious notation, we must have
\begin{equation}\label{eq:critpq}
\begin{cases}
 p_1 \to \pm\infty \quad \text{and} \quad p_2 \to \mp\infty,\\
 q_1 \to q_{\infty} \quad \text{and} \quad q_2\to q_{\infty}
\end{cases}
\end{equation}
for some finite value $q_{\infty}$, as $y\to y_0$. (The condition on $q$ is obtained from \eqref{eq:critBCs}, noting the change in sign due to the different values of $k$ on the two sections of the free surface.)

It is straightforward to show that the system \eqref{eq:pqodes} has a conserved quantity
\begin{equation}\label{eq:Hdef}
 H(p,q;k) = 1-\dfrac{1}{(1+p^2)^{1/2}} -\dfrac{k}{2}q^2,
\end{equation}
to which we shall refer informally as the energy of a solution (although it does not represent a physical energy). Orbits in either the $(p,q)$ phase plane corresponding to $k=1$ or the $(p,q)$ phase plane corresponding to $k=-1$ therefore correspond to contours of $H(p,q;\pm 1)$; this will be crucial to the construction.

\begin{figure}[tbp]
\begin{center}
\setlength{\unitlength}{1.0cm}

 \begin{picture}(12,10.5)

 \put(1.0,0.0){\epsfig{figure={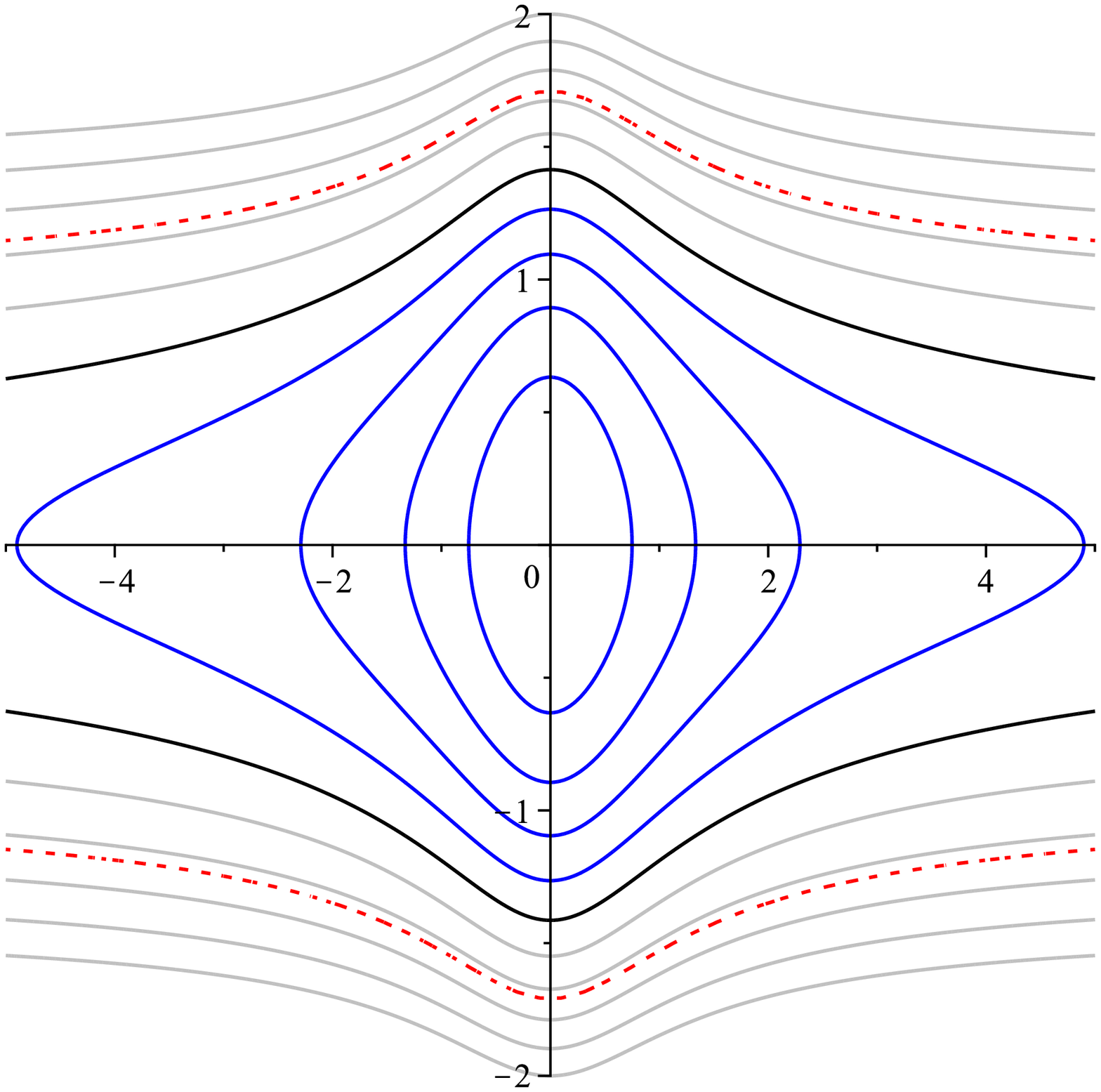},width=10\unitlength}}

 \put(11.0,4.6){\footnotesize{$p$}}
 \put(6.2,9.8){\footnotesize{$q$}}

 \put(10.9,6.35){\footnotesize{$E_-=1$}}
 \put(10.9,3.4){\footnotesize{$E_-=1$}}

 \put(9.0,3.5){\vector(3,1){0.5}}
 \put(3.0,3.5){\vector(3,-1){0.5}}
 \put(9.0,6.5){\vector(-3,1){0.5}}
 \put(3.0,6.5){\vector(-3,-1){0.5}}

 \end{picture}
 \caption{Phase portrait of \eqref{eq:pqodes} with $k=-1$; orbits correspond to $H(p,q,;-1)=E_-$. Closed orbits (blue) represent energies $0<E_-<1$; unbounded orbits (gray) represent energies $E_->1$; the black orbit represents $E_-=1$; the dashed red orbit represents $E_-=E^*\approx 1.462$ (see \S\ref{subsec:largeenergy}). Arrows indicate the direction of travel in each quadrant as $y$ increases.}
  \label{fig:phaseplane_k=-1}

\end{center}
\end{figure}

\begin{figure}[tbp]
\begin{center}
\setlength{\unitlength}{1.0cm}

 \begin{picture}(12,10.5)

 \put(1.0,0.0){\epsfig{figure={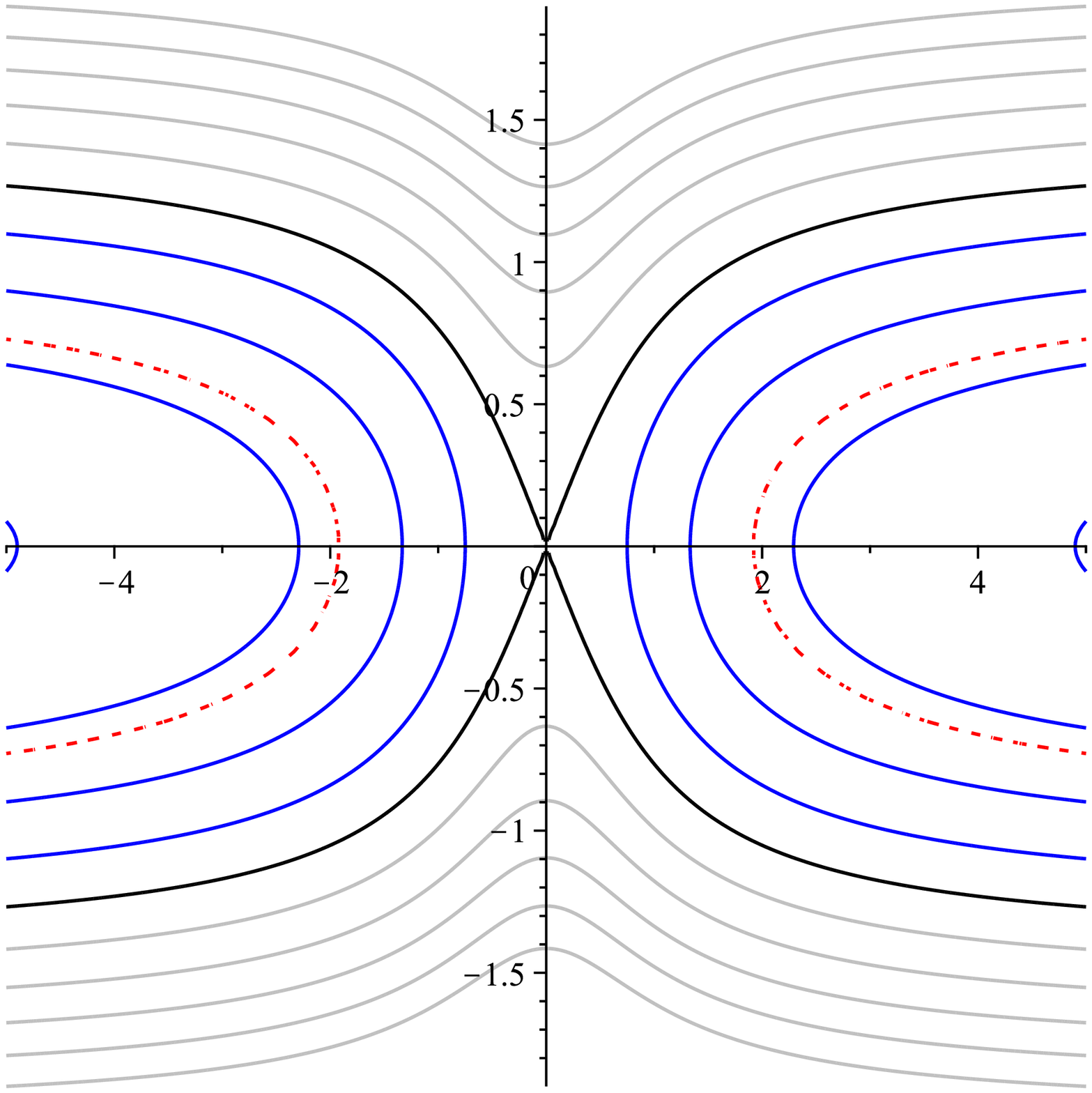},width=10\unitlength}}

 \put(11.0,4.6){\footnotesize{$p$}}
 \put(6.2,9.8){\footnotesize{$q$}}

 \put(10.9,8.15){\footnotesize{$E_+=0$}}
 \put(10.9,1.6){\footnotesize{$E_+=0$}}

 \put(9.4,3.2){\vector(-3,1){0.5}}
 \put(3.4,3.4){\vector(-3,-1){0.5}}
 \put(9.0,6.65){\vector(3,1){0.5}}
 \put(3.0,6.65){\vector(3,-1){0.5}}

 \end{picture}
 \caption{Phase portrait of \eqref{eq:pqodes} with $k=+1$; orbits correspond to $H(p,q,;+1)=E_+$. $p$-bounded orbits (blue) represent energies $0<E_+<1$; $q-$bounded orbits (gray) represent energies $E_+<0$; the separatrix (black) represents $E_+=0$; the dashed red orbit represents $E_+=2-E^*\approx 0.538$ (see \S\ref{subsec:largeenergy}). Arrows indicate the direction of travel in each quadrant as $y$ increases.}
  \label{fig:phaseplane_k=+1}

\end{center}
\end{figure}

Figure \ref{fig:phaseplane_k=-1} illustrates the phase plane for $k=-1$. By inspection of the equation $H(p,q;-1)=E_-$ it is clear that closed orbits (blue in the figure) correspond to $0<E_-<1$; such orbits intersect the $q$-axis at $q=\pm q_0(E_-)$ and intersect the $p$-axis at $p=\pm p_0(E_-)$, where
\begin{equation}
 q_0(E)=\sqrt{2E} \qquad \text{and} \qquad p_0(E) = \left[\dfrac{2E-E^2}{1-E}\right]^{1/2}.
\end{equation}
Similarly, unbounded orbits (gray in the figure) correspond to $E_->1$; these intersect the $q$-axis at $q=\pm q_0(E_-)$ and satisfy $|q|\to q^-_{\infty}(E_-)$ as $|p|\to\infty$, where
\begin{equation}
 q^-_{\infty}(E) = \sqrt{2(E-1)}.
\end{equation}
The critical orbit corresponding to $E_-=1$ is shown in black. We will discuss the significance of the orbit corresponding to $E_-=E^*$ (shown in dashed red) in \S\ref{subsec:largeenergy}.

Figure \ref{fig:phaseplane_k=+1} illustrates the phase plane for $k=+1$. By inspection of the equation $H(p,q;+1)=E_+$ it is clear that orbits that are bounded in $p$ (blue in the figure) correspond to $0<E_+<1$; such orbits intersect the $p$-axis at $p=\pm p_0(E_+)$. Similarly, orbits that are bounded in $q$ (gray in the figure) correspond to $E_+<0$; these intersect the $q$-axis at $q=\pm q_0(E_+)$. All orbits satisfy $|q|\to q^+_{\infty}(E_+)$ as $|p|\to\infty$, where
\begin{equation}
 q^+_{\infty}(E) = \sqrt{2(1-E)}.
\end{equation}
The separatrix corresponding to $E_+=0$ is shown in black. We will discuss the significance of the orbit corresponding to $E_+=2-E^*$ (shown in dashed red) in \S\ref{subsec:largeenergy}.

\subsection{Time maps\label{subsec:timemaps}}

We can rearrange the equation $H(p,q;k)=E$ to obtain
\begin{equation}\label{eq:dqdy}
 \dfrac{\d q}{\d y} = \pm \left[ \dfrac{1}{(1-E-\frac{1}{2}kq^2)^2} - 1 \right]^{1/2}.
\end{equation}
The `time' taken to traverse an orbit from $q=q_1$ to $q=q_2>q_1$ (in physical terms, the lateral distance between these two points) is therefore given by
\begin{equation}\label{eq:Lq1q2}
 L(q_1,q_2;E,k) = \int_{q_1}^{q_2}\left[ \dfrac{1}{(1-E-\frac{1}{2}kq^2)^2} - 1 \right]^{-1/2}\d q = \int_{q_1}^{q_2} \dfrac{2(1-E)-kq^2}{\left[(2E+kq^2)(4-2E-kq^2)\right]^{1/2}}\d q.
\end{equation}

\subsection{Construction of pendent solutions with small energy\label{subsec:smallenergy}}

We first consider the construction of solutions for pendent rivulets with small energy, for which the liquid lies above the air at every point on the free surface ($k=-1$). Such a solution corresponds to a closed orbit in the $k=-1$ phase plane, which we characterise by $H(p,q;-1) = E_-$, where $0<E_-<1$. As we have seen, the orbit intersects the $q$-axis at $q=\pm q_0(E_-)$ and the $p$-axis at $p=\pm p_0(E_-)$.

\begin{remark}\label{remark:restrictedsols}
Although we restrict ourselves here to solutions that traverse the orbit once, it is clear that we can construct solutions that traverse it any whole number of times, corresponding to $n\geq 2$ in the lubrication solution \eqref{eq:zetasol_lub}. We also restrict ourselves to physically meaningful solutions $\zeta\geq 0$; each such solution has a counterpart unphysical solution $\zeta\leq 0$, which is constructed by starting at $q=+q_0(E_-)$.
\end{remark}

The half-width of the rivulet is given from \eqref{eq:Lq1q2} by
\begin{align}\label{eq:alowE}
 a(E_-) & = L(-q_0(E_-),q_0(E_-);E_-,-1)
 = \int_{-\sqrt{2E_-}}^{\sqrt{2E_-}} \dfrac{2(1-E_-)+q^2}{\left[(2E_--q^2)(4-2E_-+q^2)\right]^{1/2}}\d q.
\end{align}

\begin{prop}\label{prop:alowEmonotonic}
The half-width $a(E_-)$ given by \eqref{eq:alowE} is a monotonically decreasing function of $E_-$ on $0<E_-\leq1$.
\end{prop}

\begin{proof}
By the change of variables $q^2=2E_-u$ we can write
\begin{equation}
 a(E_-) = \sqrt{2}\int_0^1 \dfrac{1-(1-u)E_-}{[u(1-u)(2-(1-u)E_-)]^{1/2}}\d u.
\end{equation}
It follows that
\begin{equation}
 \dfrac{\d}{\d E_-}a(E_-) = \dfrac{1}{\sqrt{2}}\int_0^1 \dfrac{-(1-u)^{1/2}(3-(1-u)E_-)}{[u(2-(1-u)E_-)^3]^{1/2}}\d u.
\end{equation}
Recalling that $0<E_-\leq 1$, it is clear that the integrand is negative and thus $a(E_-)$ given by \eqref{eq:alowE} is monotonically decreasing as required.
\end{proof}

The maximum gradient occurs when $q=0$, and is given by $p_0(E_-)$; as $E_-\to 1^-$, $p_0(E_-)\to\infty$. The critical orbit $E_-=1$ thus corresponds to a rivulet profile the gradient of which becomes infinite at $y=\pm\frac{1}{2}a(1)$. We may evaluate the half-width of this critical solution as
\begin{equation}
 a(1) = \int_{-\sqrt{2}}^{\sqrt{2}} \dfrac{q^2}{\left[4-q^4\right]^{1/2}}\d q = 4E\left(\frac{1}{\sqrt{2}}\right) -2K\left(\frac{1}{\sqrt{2}}\right) \approx 1.694,
\end{equation}
where $K$ and $E$ are the complete elliptic integrals of the first and second kind.

Finally, we impose the boundary conditions $\zeta(\pm a) = 0$ by setting $\zeta(y)=q(y)+q_0(E_-)$. It follows that the maximum thickness of the rivulet is $h_{\mathrm{m}} = 2q_0(E_-) = 2\sqrt{2E_-}$. Further, the symmetry of the orbits in $q$ means that we can easily calculate the half-area
\begin{equation}\label{eq:arealowE}
 A(E_-) = \int_{-a(E_-)}^0 \zeta(y)\,\d y = \int_{-a(E_-)}^{0}\left(q(y)+q_0(E_-)\right)\,\d y = q_0(E_-)a(E_-).
\end{equation}

\begin{figure}[tbp]
\begin{center}
\setlength{\unitlength}{1.0cm}

 \begin{picture}(14,7)

 \put(0.0,0.0){\epsfig{figure={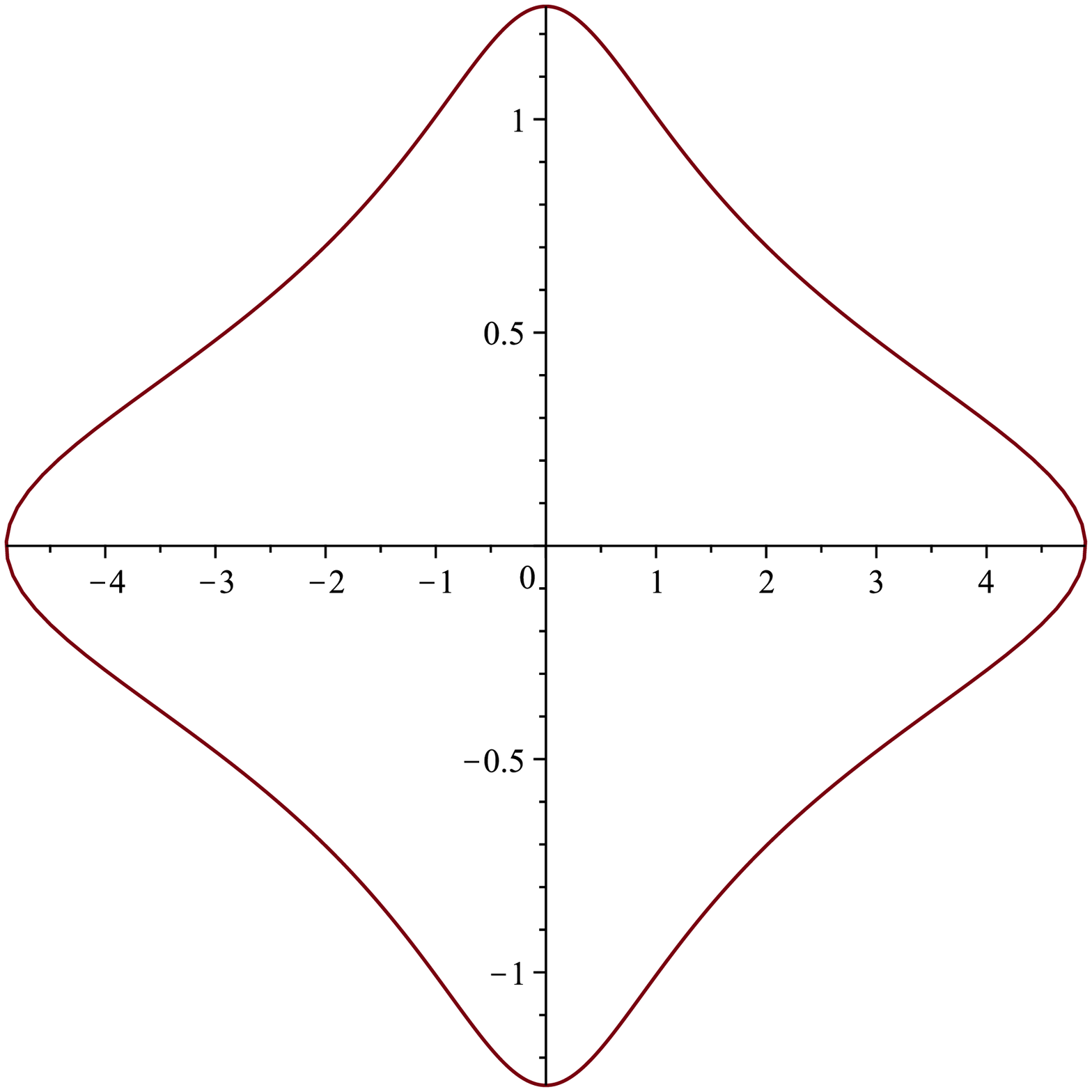},width=7\unitlength}}
 \put(7.0,0.0){\epsfig{figure={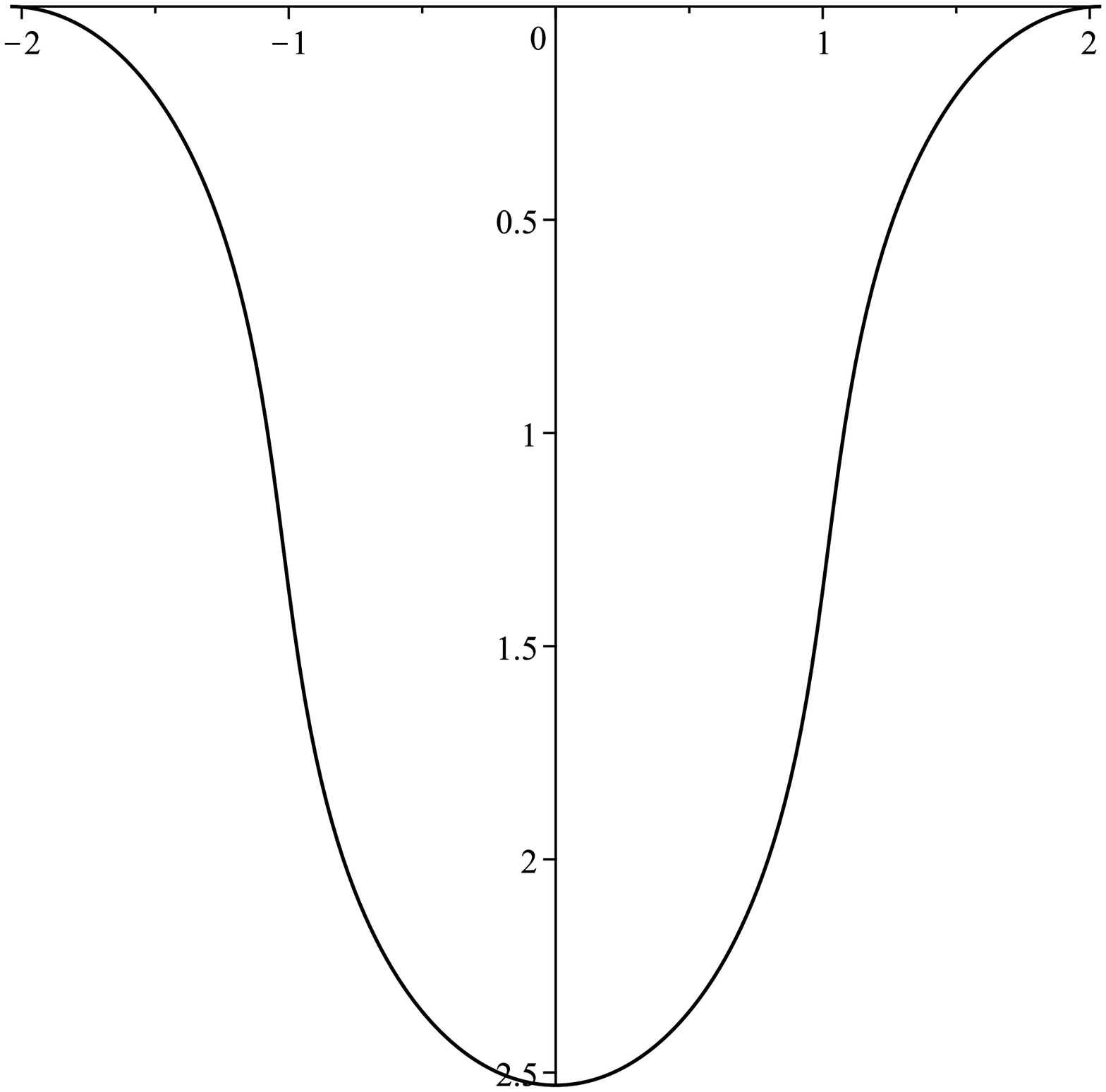},width=7\unitlength}}

 \put(0.0,6.5){\footnotesize{(a)}}
 \put(6.5,6.5){\footnotesize{(b)}}

 \put(5.0,3.0){\scriptsize{$p$}}
 \put(3.7,5.2){\scriptsize{$q$}}

 \put(10.75,3.3){\scriptsize{$z$}}
 \put(12.5,7.0){\scriptsize{$y$}}

 \put(5.2,1.3){\vector(1,1){0.5}}
 \put(1.3,1.8){\vector(1,-1){0.5}}
 \put(5.7,5.1){\vector(-1,1){0.5}}
 \put(1.8,5.6){\vector(-1,-1){0.5}}

 \end{picture}
 \caption{The solution corresponding to $E_-=0.8$. (a) The orbit in the $(p,q)$ phase plane; arrows indicate the direction of increasing $y$. (b) The rivulet profile $z=\zeta(y)$ (note that the $z$-coordinate increases downward). We have $q_0(E_-)\approx 1.265$, $a(E_-)\approx 2.042$, $A(E_-)\approx 2.584$.}
 \label{fig:E=0.8}

\end{center}
\end{figure}

Figure \ref{fig:E=0.8} illustrates a typical solution. The orbit in the phase plane is traversed anticlockwise from $(0,-q_0)$ and the profile is traversed from $y=-a$ to $y=+a$. Note that the symmetry of the orbit under reflection in the $p$- and $q$-axes corresponds to symmetry between the upper and lower parts, as well as the left- and right-hand parts, of the profile.

\subsection{Construction of pendent solutions with larger energy\label{subsec:largeenergy}}

Orbits in the $k=-1$ phase plane corresponding to energies $E_->1$ do not connect two points on the $q$-axis and therefore cannot satisfy the boundary-value problem \eqref{eq:pqodes}, \eqref{eq:pbc}. Instead, the argument of \S\ref{sec:goveqs} leading to \eqref{eq:C} shows that we must construct an `overhanging' rivulet profile by connecting sections of the free surface on which liquid lies above air ($k=-1$) with sections on which liquid lies below air ($k=+1$). These sections must be connected at critical points at which the free surface becomes vertical and the conditions \eqref{eq:critBCs} (and thus \eqref{eq:critpq}) apply.

We describe the construction of a rivulet with a single overhang between two critical points $y_1$ and $y_2$ with $y_1<y_2$; we will show (Remark \ref{rem:overhangs}) that rivulets with multiple overhangs are impossible. The section of the free surface between the contact line $y=-a$ and the first critical point $y=y_2$ is given by $z=\zeta_1(y)$; on this section, liquid lies above air, $k=-1$. The section between the critical points $y=y_2$ and $y=y_1$ is given by $z=\zeta_2(y)$; on this section, liquid lies below air, $k=+1$. The section between the critical point $y=y_1$ and the midline of the rivulet $y=0$ is given by $z=\zeta_3(y)$; on this section, liquid again lies above air, $k=-1$. (The reader may find it helpful to refer to figure \ref{fig:E=1.3}.) Once we have constructed the free surface between $y=-a$ and $y=0$, the other half follows by symmetry.

The construction proceeds as follows, for an energy $E_- > 1$. 
\begin{itemize}
 \item[(i)] To construct $\zeta_1(y)$ between the contact line $y=-a$ and the critical point $y=y_2$, we start at $(p,q)=(0,-q_0(E_-))$ in the $k=-1$ phase plane, and follow this orbit in the direction of increasing $y$ until $(p,q)\to (\infty,-q^-_{\infty}(E_-))$.
 \item[(ii)] By \eqref{eq:critpq}, this orbit now connects to the orbit in the $k=+1$ phase plane which satisfies $(p,q) \to (-\infty,-q^-_{\infty}(E_-))$. Setting $q^-_{\infty}(E_-)=q^+_{\infty}(E_+)$, we determine that the corresponding energy in the $k=+1$ phase plane is $E_+=2-E_-$.
 \item[(iii)] To construct $\zeta_2(y)$ between the critical points $y=y_2$ and $y=y_1$, we follow this orbit in the $k=+1$ phase plane in the direction of \emph{decreasing} $y$ from $(p,q)\to(-\infty,-q^+_{\infty}(E_+))$ to $(p,q) \to (-\infty,q^+_{\infty}(E_+))$.
 \item[(iv)] By \eqref{eq:critpq}, this orbit now connects to the orbit in the $k=-1$ phase plane which satisfies $(p,q) \to (\infty,q^-_{\infty}(E_-))$.
 \item[(v)] To construct $\zeta_3(y)$ between the critical point $y=y_1$ and the midline $y=0$, we follow this orbit in the $k=-1$ phase plane in the direction of increasing $y$ from $(p,q)\to (\infty,q^-_{\infty}(E_+))$ to $(p,q)=(0,q_0(E_-))$.
  \item[(vi)] The rest of the construction follows by symmetry.
\end{itemize}

To determine the critical points $y_1$ and $y_2$ and the half-width $a$ we again use the `time map' \eqref{eq:Lq1q2}. From the construction described above, we have
\begin{align}
 y_2-(-a) & = L_-(E_-) = L(-q_0(E_-),-q^-_{\infty}(E_-);E_-,-1),\\
 y_2-y_1 & = L_+(E_-) = L(-q^-_{\infty}(E_-),q^-_{\infty}(E_-);2-E_-,+1),\\
 0-y_1 & = L(q^-_{\infty}(E_-),q_0(E_-);E_-,-1) = L_-(E_-).
\end{align}
Thus
\begin{equation}\label{eq:y1y1highE}
 y_1 = -L_-(E_-), \qquad y_2 =-L_-(E_-)+L_+(E_-),
\end{equation}
and
\begin{equation}\label{eq:ahighE}
 a(E_-) = 2L_-(E_-) - L_+(E_-).
\end{equation}

\begin{prop}\label{prop:ahighEmonotonic}
The half-width $a(E_-)$ given by \eqref{eq:ahighE} is a monotonically decreasing function of $E_-$ on $1\leq E_-<2$.
\end{prop}

\begin{proof}
We first show that $L_-(E_-)$ is a monotonically decreasing function of $E_-$. By the change of variables $q^2=2E_-u$, we can write
\begin{equation}
 L_-(E_-) = \dfrac{1}{\sqrt{2}}\int_{1-1/E_-}^1\dfrac{1-(1-u)E_-}{[u(1-u)(2-(1-u)E_-)]^{1/2}}\d u.
\end{equation}
It follows after a little algebraic manipulation that
\begin{equation}
 \dfrac{\d L_-}{\d E_-} = \dfrac{1}{2\sqrt{2}}\int_{1-1/E_-}^1 \dfrac{-(3-(1-u)E_-)(1-u)}{[u(1-u)(2-(1-u)E_-)]^{1/2}(2-(1-u)E_-)}\d u.
\end{equation}
Since $1\leq E_-< 2$ and $0<u\leq1$, it is clear that the integrand is negative and thus $L_-(E_-)$ is monotonically decreasing in $E_-$ as required.

We now show that $L_+(E_-)$ is a monotonically increasing function of $E_-$. By the change of variables $q^2=2(E_--1)u$, we can write
\begin{equation}
 L_+(E_-) = 2(E_--1)^{3/2}\int_0^1\dfrac{1-u}{[u(E_--(E_--1)u)(2-E_-+(E_--1)u)]^{1/2}}\d u.
\end{equation}
Now, $(E_--1)^{3/2}$ is a monotonically increasing function of $E_-$, and it follows after a little more algebraic manipulation that
\begin{multline}
 \dfrac{\d}{\d E_-}\int_0^1\dfrac{1-u}{[u(E_--(E_--1)u)(2-E_-+(E_--1)u)]^{1/2}}\d u = \\
 \sqrt{2}\int_0^1\dfrac{(E_--1)(1-u)^3\d u}{[u(2-E_-+(E_--1)u)(u+(1-u)E_-)]^{1/2}(E_--(E_--1)u)(2-E_-+(E_--1)u)}.
\end{multline}
Since $1\leq E_-< 2$ and $0<u\leq1$, it is clear that the integrand is positive and thus $L_+(E_-)$ is monotonically increasing in $E_-$ as required.

It follows that the half-width $a(E_-)$ given by \eqref{eq:ahighE} is a monotonically decreasing function of $E_-$ as required.
\end{proof}

The half-area $A$ is again straightforward to calculate, because the symmetry of the profile leads to many cancellations; it can readily be shown that it is given by
\begin{align}\label{eq:areahighE}
 A(E_-) & = \int_{y_1}^0(\zeta_3+q_0)\,\d y -\int_{y_1}^{y_2}(\zeta_2+q_0)\,\d y + \int_{-a}^{y_2}(\zeta_1+q_0)\,\d y \nonumber\\
 & = q_0(E_-)a(E_-)
\end{align}
as for the low-energy solution. We will plot the half-area against the half-width in \S\ref{sec:bifurcation}.

\begin{figure}[tbp]
\begin{center}
\setlength{\unitlength}{1.0cm}

 \begin{picture}(14,7)

 \put(0.0,0.0){\epsfig{figure={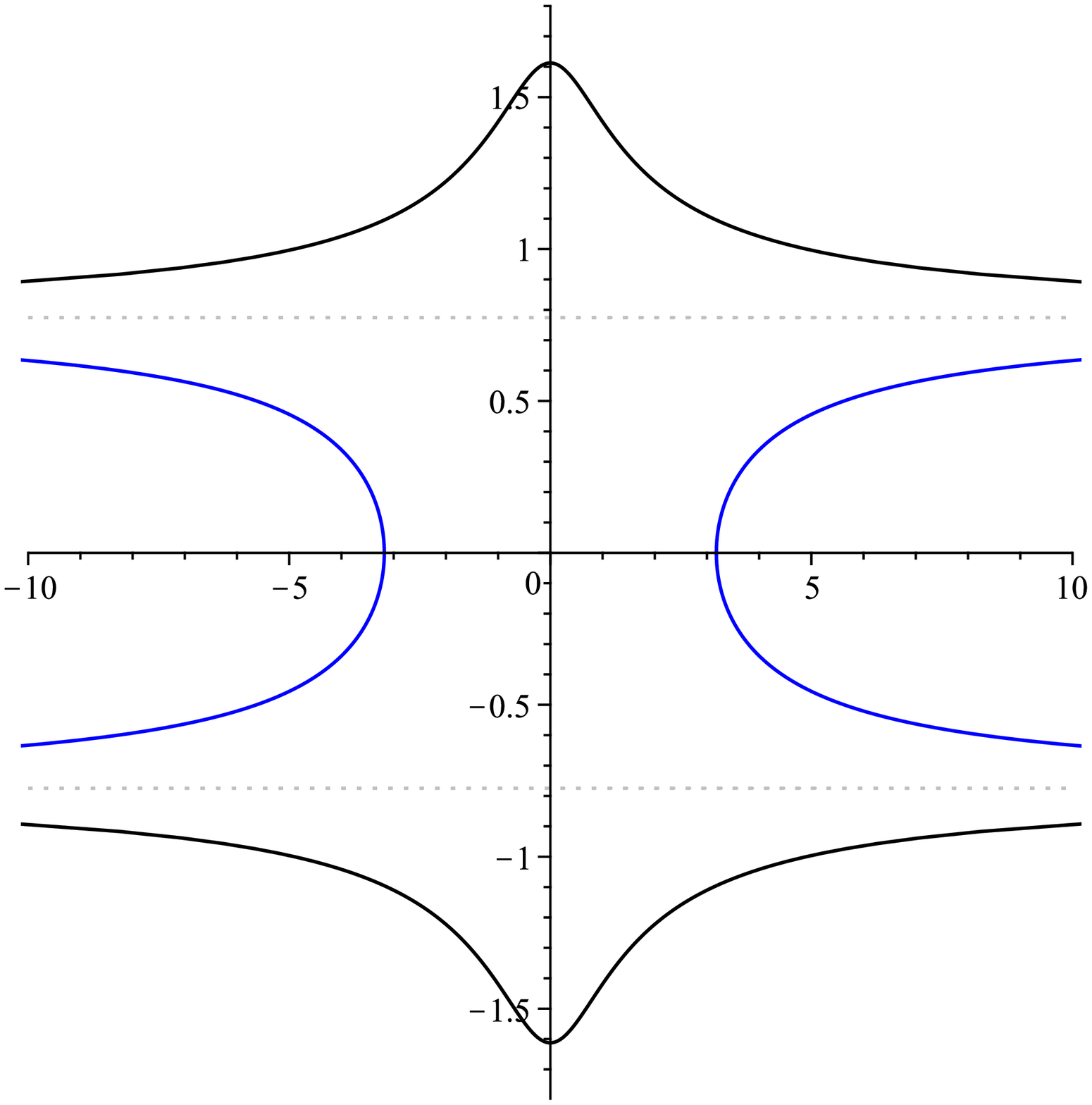},width=7\unitlength}}
 \put(7.0,0.0){\epsfig{figure={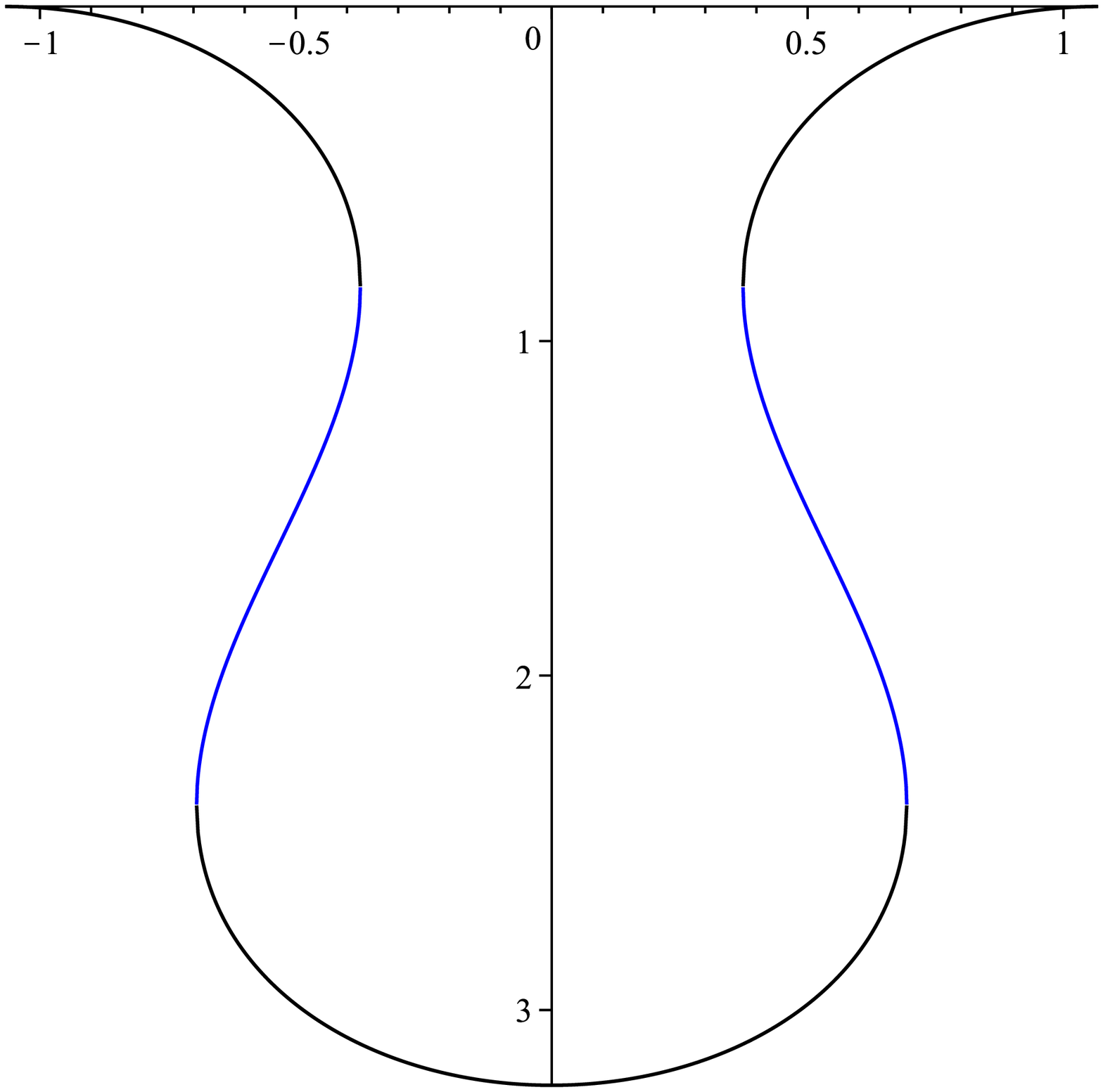},width=7\unitlength}}

 \put(0.0,6.5){\footnotesize{(a)}}
 \put(6.5,6.5){\footnotesize{(b)}}

 \put(5.0,3.0){\scriptsize{$p$}}
 \put(3.7,5.2){\scriptsize{$q$}}

 \put(10.75,3.3){\scriptsize{$z$}}
 \put(12.5,7.0){\scriptsize{$y$}}

 \put(5.2,1.3){\vector(3,1){0.5}}
 \put(1.3,1.5){\vector(3,-1){0.5}}
 \put(5.7,5.45){\vector(-3,1){0.5}}
 \put(1.8,5.6){\vector(-3,-1){0.5}}

 \put(5.2,2.35){\color{blue}{\vector(3,-1){0.5}}}
 \put(5.7,4.75){\color{blue}{\vector(-3,-1){0.5}}}
 \put(1.3,2.2){\color{blue}{\vector(3,1){0.5}}}
 \put(1.8,4.6){\color{blue}{\vector(-3,1){0.5}}}

 \put(8.0,6.0){\scriptsize{$z=\zeta_1$}}
 \put(8.1,3.8){\color{blue}{\scriptsize{$z=\zeta_2$}}}
 \put(8.0,0.5){\scriptsize{$z=\zeta_3$}}

 \put(4.0,0.5){\scriptsize{$k=-1$}}
 \put(4.0,6.4){\scriptsize{$k=-1$}}
 \put(1.3,3.9){\color{blue}{\scriptsize{$k=+1$}}}
 \put(5.0,3.9){\color{blue}{\scriptsize{$k=+1$}}}
 
 \end{picture}
 \caption{The solution corresponding to $E_-=1.3$. (a) The orbit in the $(p,q)$ phase plane. (b) The rivulet profile made up of the sections $z=\zeta_1(y)$, $z=\zeta_2(y)$, and $z=\zeta_3(y)$ (note that the $z$-coordinate increases downward). Sections and arrows in black correspond to the $k=-1$ phase plane; arrows show $y$ increasing. Sections and arrows in blue correspond to the $k=+1$ phase plane; arrows show $y$ decreasing. The dotted gray lines in (a) are $q=\pm q^-_{\infty}(E_-)$. We have $q_0(E_-)\approx 1.612$, $q^-_{\infty}(E_-) \approx 0.775$, $a(E_-)\approx 1.068$, $A(E_-)\approx 1.722$.}
 \label{fig:E=1.3}

\end{center}
\end{figure}

Figure \ref{fig:E=1.3} illustrates a typical solution. The orbit starts in the $k=-1$ phase plane at $(0,-q_0(E_-)$ and follows the construction described above. Note the symmetry of the profile within the middle section (blue in the figure) and between the upper and lower sections (black in the figure).

\begin{remark}\label{rem:overhangs}
It is clear from the construction described above that the solution must return to $(p,q)=(0,-q_0(E_-))$ after having fully traversed two unbounded orbits in the $k=-1$ phase plane and two in the $k=+1$ phase plane. This demonstrates that it is not possible to construct solutions with more than a single overhang.
\end{remark}

\subsection{Pinch-off\label{sec:pinchoff}}

From Proposition \ref{prop:ahighEmonotonic} we see that as the energy $E_-$ increases, the extent of the overhanging (blue) region increases and the overall width of the rivulet decreases. As a consequence, the profile eventually `pinches off' when the first transition point occurs at $y=0$; the self-intersection occurs at some $z=\zeta_{\mathrm{p}}\ > 0$. Pinch-off occurs when $E_-=E^*\in (1,2)$ such that
\begin{equation}\label{eq:Estar}
 L_-(E^*) = L_+(E^*).
\end{equation}

\begin{prop}
The value of $E^*$ is unique.
\end{prop}

\begin{Proof}
Since (from the proof of Proposition \ref{prop:ahighEmonotonic}), $L_-(E^*)$ is monotonically decreasing in $E^*$ and $L_+(E^*)$ is monotonically increasing in $E^*$, the value of $E^*$ that satisfies \eqref{eq:Estar} is unique. We may determine the value of $E^*$ numerically as $E^*\approx 1.462$.
\end{Proof}

For values of $E>E^*$, rivulet profiles can still be constructed as described in \S\ref{subsec:largeenergy}, but they are now self-intersecting and do not describe a physically realisable free surface. We can still define the half-area of a self-intersecting solution, and will do so when we construct bifurcation diagrams in \S\ref{sec:bifurcation}. However, it is important to note that self-intersection is a topological change which is not captured when the bifurcation diagrams are plotted in terms of half-area.

\begin{prop}\label{prop:pinchoffarea}
For the pinch-off solution with $E=E^*$, the region enclosed by the free surface for $z > \zeta_{\mathrm{p}}$ has half-area $1$.
\end{prop}

\begin{Proof}
By construction, the half-area of the region concerned is given by
\begin{equation}
 A_{\mathrm{p}} = \int_{-a}^0\zeta_3(y)\,\d y - \int_{-a}^0\zeta_2(y)\,\d y.
\end{equation}
Using \eqref{eq:dqdy}, we can rewrite this as
\begin{align}
 A_{\mathrm{p}} & = \int_{q^-_{\infty}(E^*)}^{q_0(E^*)} \dfrac{q\,\d q}{\left[ \frac{1}{(1-E^*+\frac{1}{2}q^2)^2} - 1 \right]^{1/2}} +\int_{-q^+_{\infty}(2-E^*)}^{q^+_{\infty}(2-E^*)} \dfrac{q\,\d q}{\left[ \frac{1}{(1-(2-E^*)-\frac{1}{2}q^2)^2} - 1 \right]^{1/2}} \\
 {} & = \int_{\sqrt{2(E^*-1)}}^{\sqrt{2E^*}} \dfrac{q\,\d q}{\left[ \frac{1}{(1-E^*+\frac{1}{2}q^2)^2} - 1 \right]^{1/2}} +\int_{-\sqrt{2(E^*-1)}}^{\sqrt{2(E^*-1)}} \dfrac{q\,\d q}{\left[ \frac{1}{(1-(2-E^*)-\frac{1}{2}q^2)^2} - 1 \right]^{1/2}} \nonumber\\
 {} & = \int_{-\sqrt{2(E^*-1)}}^{\sqrt{2E^*}} \dfrac{q\,\d q}{\left[ \frac{1}{(1-E^*+\frac{1}{2}q^2)^2} - 1 \right]^{1/2}} \nonumber\\
 {} & = 1 \nonumber
\end{align}
by evaluating the integral directly. (We note that this result does not rely on the numerical value of $E^*$.)

\end{Proof}

\section{Bifurcation analysis\label{sec:bifurcation}}

Although we have used $E_-$ as the control parameter when constructing solutions, it is convenient instead to regard the half-width $a$ as the parameter; by Propositions \ref{prop:alowEmonotonic} and \ref{prop:ahighEmonotonic} we know that $a$ is a monotonically decreasing function of $E_-$. To characterise a solution we consider the half-area $A$; for $0\leq E_-\leq 1$ this can be thought of as the $L^1$-norm of the solution $\zeta(y)$, while for $E_->1$ we may interpret it as the $L^1$-norm of a suitably defined (discontinuous) weak solution constructed using an equal-area rule. (For such a construction in an analogous problem, see \cite{BurnsGrinfeld2011}.)

Figure \ref{fig:areavswidth} shows the half-width $A$ plotted against $a$ for $0\leq E_- \leq 1.7$. The following propositions establish the crucial features of this plot.

\begin{figure}[tbp]
\begin{center}
\setlength{\unitlength}{1.0cm}

 \begin{picture}(12,10.5)

 \put(1.0,0.0){\epsfig{figure={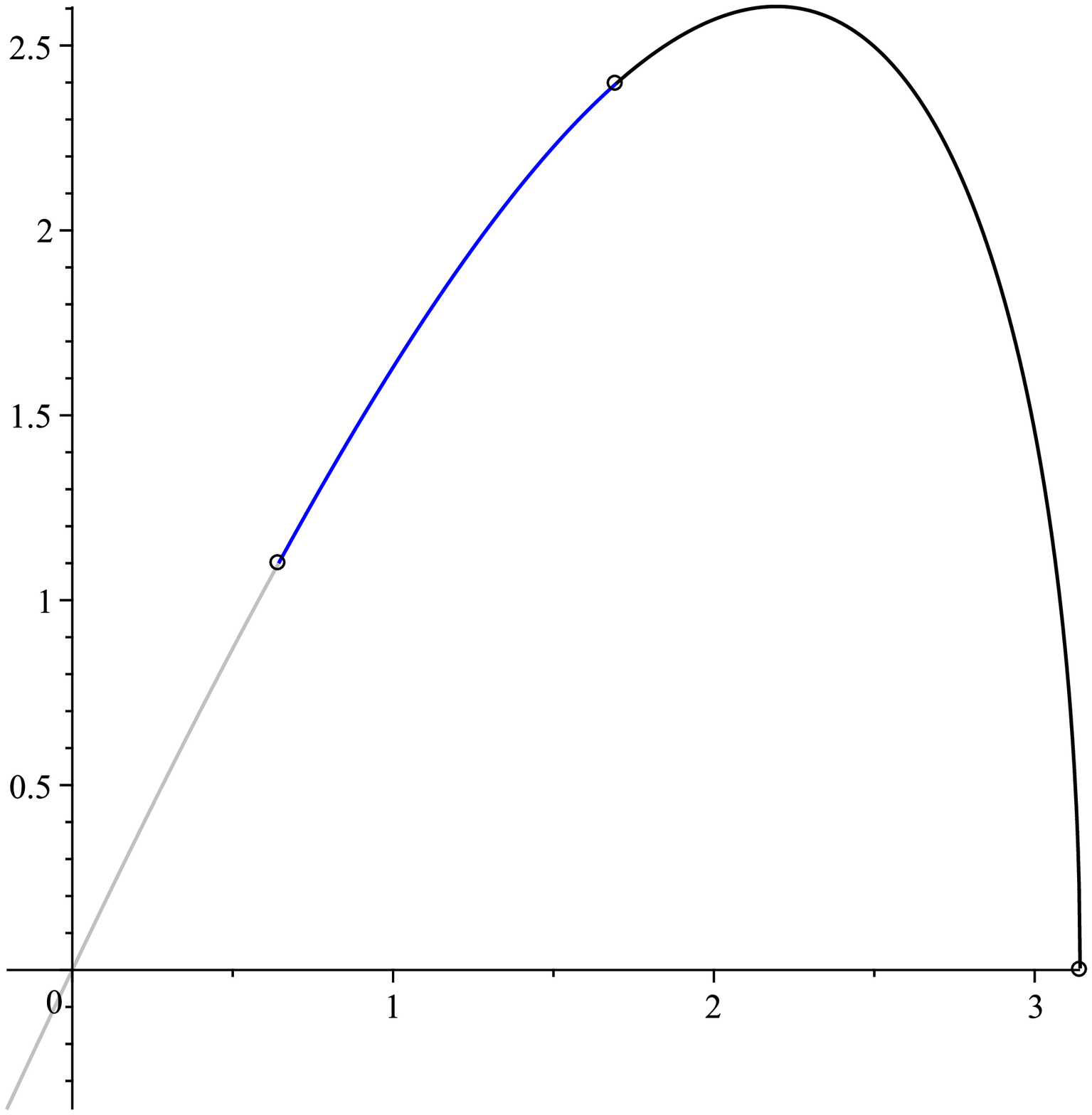},width=10\unitlength}}

 \put(6.1,0.6){\footnotesize{$a$}}
 \put(1.0,5.3){\footnotesize{$A$}}

 \put(10.9,1.3){\footnotesize{$E_-=0$}}
 \put(6.75,8.85){\footnotesize{$E_-=1$}}
 \put(3.85,4.75){\footnotesize{$E_-=E^*$}}

 \put(11.0,4.0){\epsfig{figure={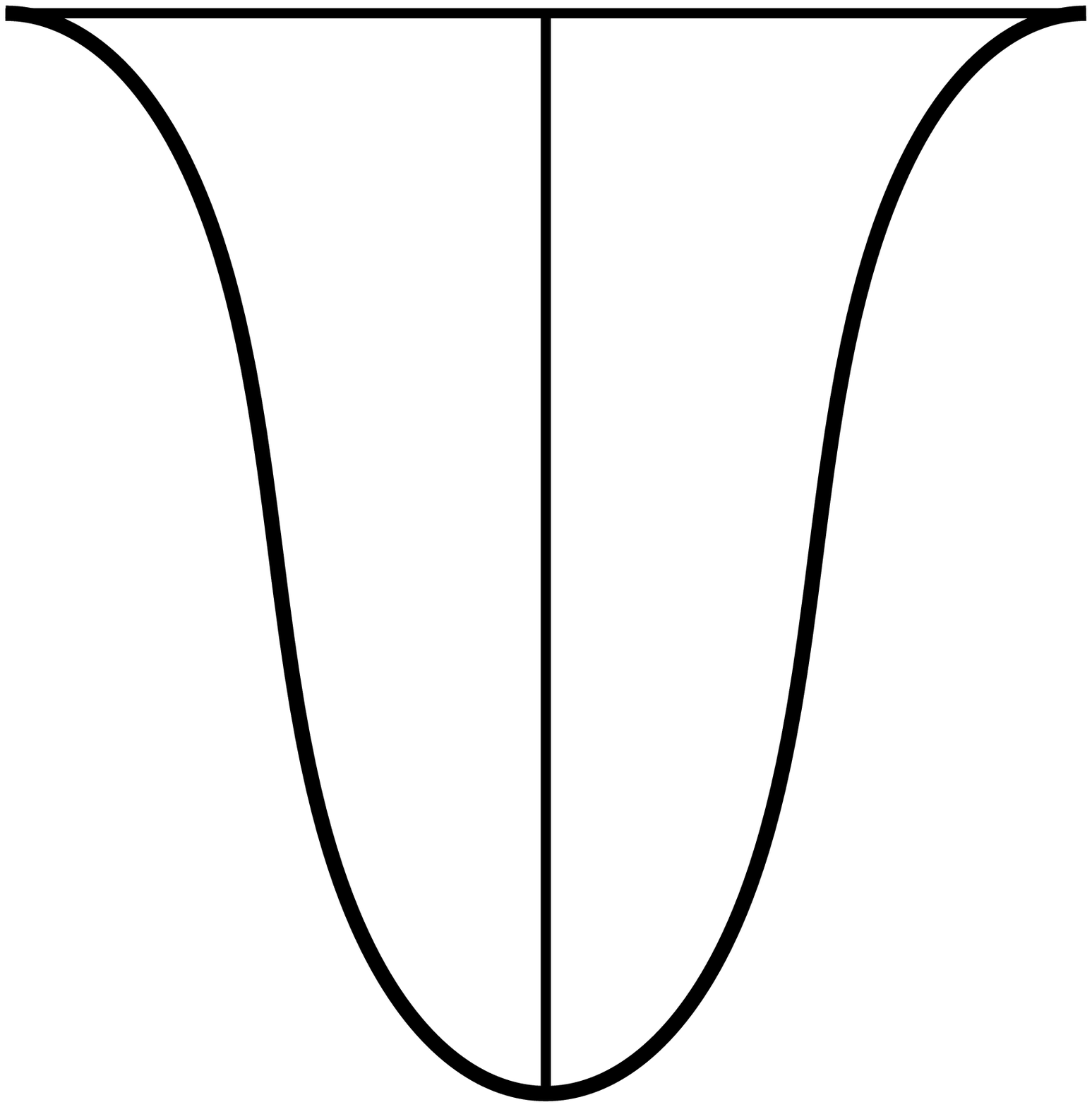},width=2\unitlength}}
 \put(3.7,8.0){\epsfig{figure={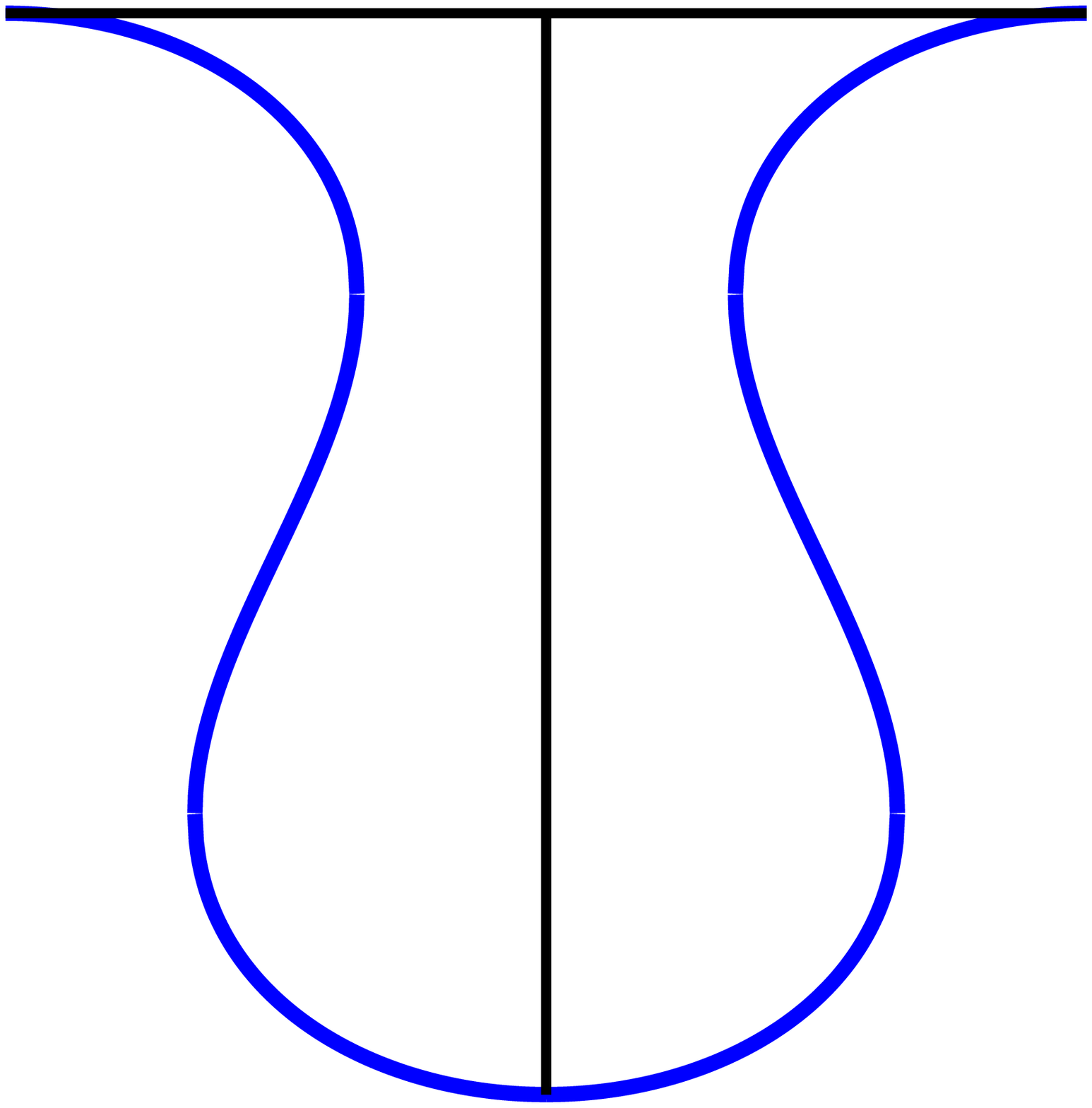},width=2\unitlength}}
 \put(3.7,2.0){\epsfig{figure={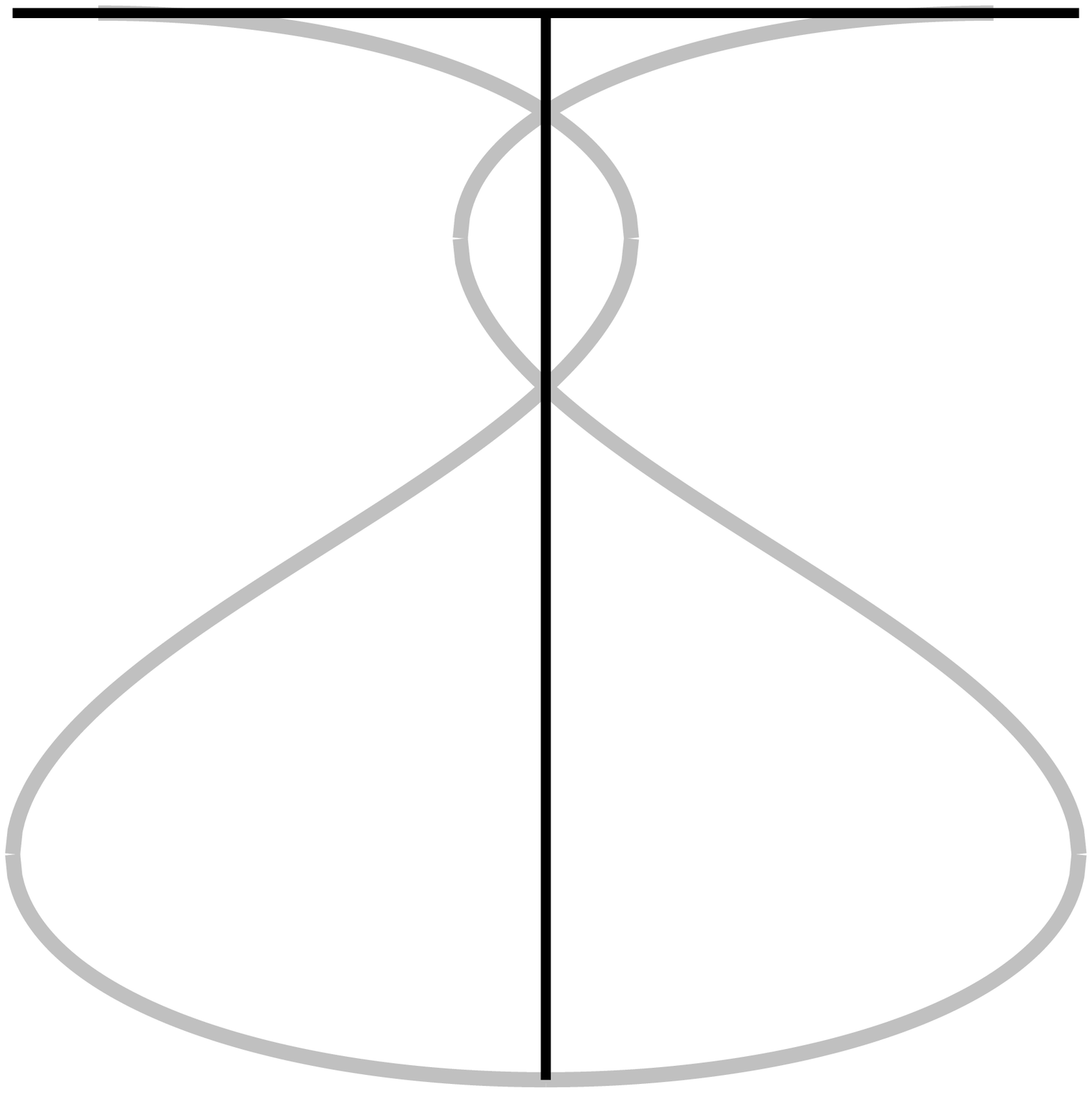},width=2\unitlength}}

 \end{picture}
 \caption{The half-area $A$ of a rivulet plotted against the half-width $a$. The black portion of the curve represents $0 < E_- < 1$ (rivulets without an overhang; the blue portion represents $1 < E_- < E^*$ (rivulets with an overhang); the gray portion represents $E_- > E^*$ (unphysical rivulets with a self-intersecting free surface). The insets (not to scale) indicate the form of the rivulet in each case.}
  \label{fig:areavswidth}

\end{center}
\end{figure}

\begin{prop}\label{prop:pitchfork}
The non-trivial solution branch emerges from the trivial solution branch $\zeta\equiv 0$ through a subcritical pitchfork bifurcation at $(a,A)=(\pi,0)$.
\end{prop}

\begin{Proof}
The bifurcation point corresponds to $E_-=0$. From \eqref{eq:alowE} we have
\begin{align}
 a(E_-) & = \int_{-\sqrt{2E_-}}^{\sqrt{2E_-}} \dfrac{2(1-E_-)+q^2}{\left[(2E_--q^2)(4-2E_-+q^2)\right]^{1/2}}\d q\\
 {} & = \sqrt{2}\int_{-1}^1\dfrac{1+(u^2-1)E_-}{[(1-u^2)(2+(u^2-1)E_-]^{1/2}}\,\d u, \label{eq:aEto0}
\end{align}
and thus
\begin{equation}
 \lim_{E_-\to0} a(E_-) = \int_{-1}^1\dfrac{\d u}{\sqrt{1-u^2}} = \pi.
\end{equation}
From \eqref{eq:arealowE} it follows that $A(E_-) = a(E_-)\sqrt{2E_-} \to0$ as $E_-\to 0$. Thus the bifurcation point is $(a,A)=(\pi,0)$ as required.

We recall from Remark \ref{remark:restrictedsols} that every physical solution with $\zeta\geq 0$ and $A> 0$ has a counterpart unphysical solution with $\zeta\leq0$ and thus $A<0$. Further, we can expand the integrand in \eqref{eq:aEto0} to obtain
\begin{equation}
 a(E_-) = \pi -\dfrac{3}{8}\pi E_- + o(E_-) \quad \text{and thus} \quad A(E_-) = \pi\sqrt{2E_-} + o(\sqrt{E_-}),
\end{equation}
and thus
\begin{equation}
 A(E_-) \sim \dfrac{4\sqrt{\pi}}{\sqrt{3}}(\pi-a)^{1/2}
\end{equation}
on the physical branch of solutions $A>0$.

From the symmetry of the solutions $A\gtrless 0$ and this local structure, it follows that the bifurcation at $a=\pi$ is a subcritical pitchfork. (This result can alternatively be obtained via Liapunov--Schmidt reduction \cite{GolubitskySchaeffer1985}.)
\end{Proof}

\begin{prop}\label{prop:origin}
The value of $E_-$ at which the non-trivial solution branch passes through $(a,A)=(0,0)$ satisfies $E_- >E^*$.
\end{prop}

\begin{Proof}
The solution branch passes through the origin when $a(E_-)=0$ and thus $A(E_-)=\sqrt{2E_-}a(E_-)=0$ [equation \eqref{eq:areahighE}]. From \eqref{eq:ahighE}, $a(E_-)=0$ when $2L_-(E_-)=L_+(E_-)$. But we know from the proof of Proposition \ref{prop:ahighEmonotonic} that $L_-(E_-)$ is monotonically decreasing and $L_+(E_-)$ is monotonically increasing, and from equation \eqref{eq:Estar} we know that the energy $E^*$ satisfies $L_-(E^*)=L_+(E^*)$. Thus $a(E_-)=0 \implies E_->E^*$.
\end{Proof}

\begin{remark}
Proposition \ref{prop:origin} means that pinch-off occurs before this solution branch reaches the origin. Numerically, we find that $a(E_-)=0$ for $E_-\approx 1.652$ (recall that $E^*\approx 1.462$).
\end{remark}

\begin{remark}
The lubrication analysis (\S\ref{sec:lubrication}) fails to capture the pitchfork bifurcation at $(a,A)=(\pi,0)$. This is not surprising, as it is clear from figure \ref{fig:phaseplane_k=-1} that the lubrication analysis corresponds to linearising about the centre $(0,0)$ in the $k=-1$ phase plane. The situation is precisely analogous to linearising a nonlinear oscillator about a non-hyperbolic equilibrium: the linearised system (simple harmonic motion) cannot predict the relationship between amplitude and period and cannot capture the other features of the phase plane, such as the separatrix. Another classical example occurs when linearising the Euler elastica around a bifurcation point \cite{Brown2014}.
\end{remark}

\section{Imperfect wetting\label{sec:imperfect}}

In this section we present illustrative solutions and bifurcation diagrams for imperfect wetting, when the contact angle $\beta > 0$. It is harder to obtain analytical results for this problem, because some of the relevant time maps are no longer monotonic in $E_-$ (cf. Proposition \ref{prop:ahighEmonotonic}). Nevertheless, all the behaviour of these imperfectly wetting solutions can be understood in terms of the perfectly wetting solutions.

Imperfectly wetting rivulets satisfy the problem \eqref{eq:zetaode}, \eqref{eq:zetabc} with $\beta\neq 0$. We will first consider cases $0<\beta<\pi/2$ (in fluid-dynamical terms, a hydrophilic substrate), and then $\beta > \pi/2$ (a hydrophobic substrate).

\subsection{Hydrophilic substrate ($0<\beta<\pi/2$)\label{sec:hydrophilic}}

\subsubsection{Construction of solutions\label{subsec:imperfect-construction}}

When $0<\beta<\pi/2$, each solution corresponds to a trajectory that starts and ends in the $k=-1$ phase plane.

\begin{figure}[tbp]
\begin{center}
\setlength{\unitlength}{1.0cm}

 \begin{picture}(12,10.5)

 \put(1.0,0.0){\epsfig{figure={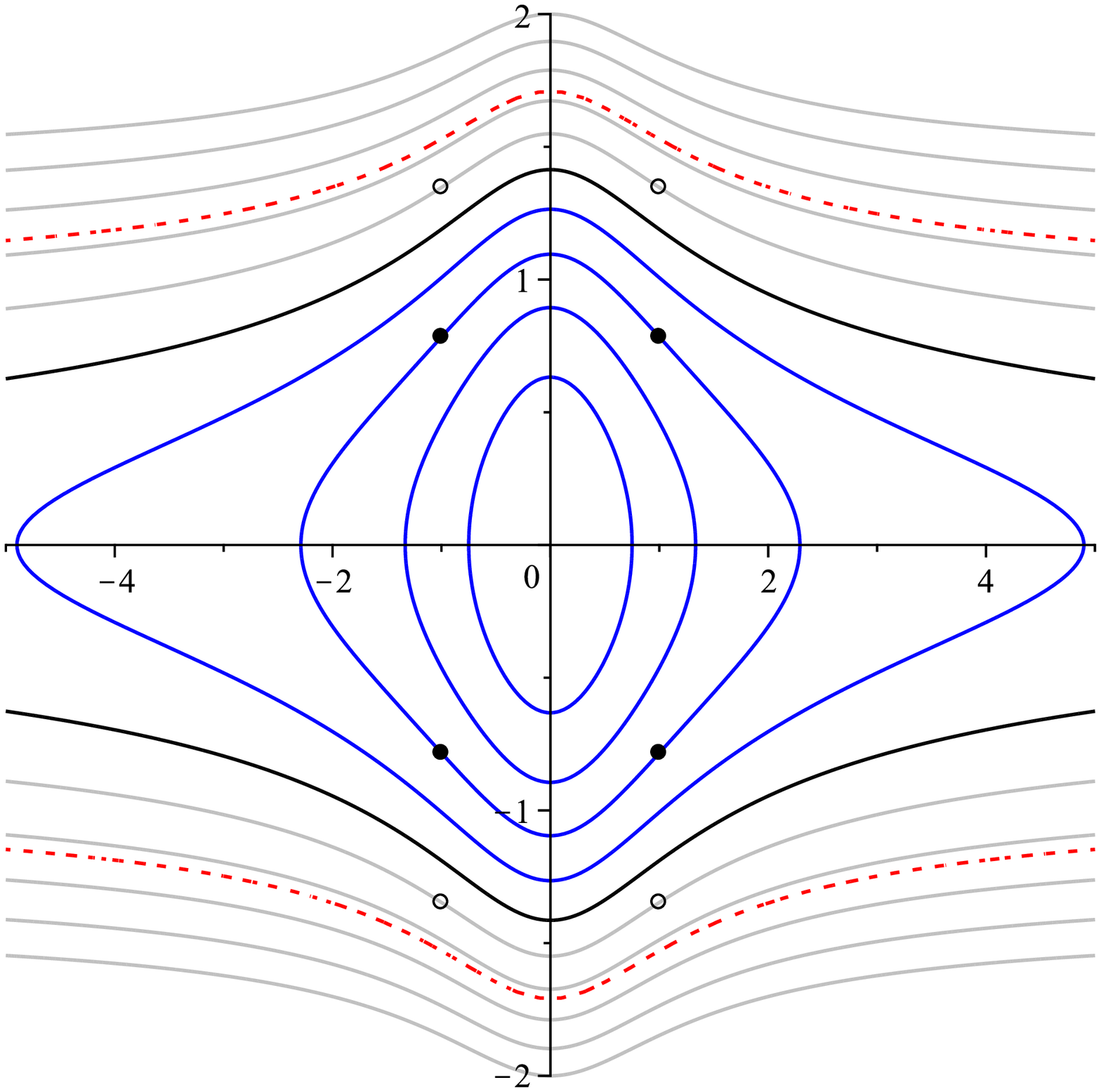},width=10\unitlength}}

 \put(11.0,4.6){\footnotesize{$p$}}
 \put(6.2,9.8){\footnotesize{$q$}}

 \put(10.9,6.35){\footnotesize{$E_-=1$}}
 \put(10.9,3.4){\footnotesize{$E_-=1$}}

 \put(9.0,3.5){\vector(3,1){0.5}}
 \put(3.0,3.5){\vector(3,-1){0.5}}
 \put(9.0,6.5){\vector(-3,1){0.5}}
 \put(3.0,6.5){\vector(-3,-1){0.5}}

 \put(7.15,6.8){\footnotesize{A}}
 \put(7.15,3.1){\footnotesize{B}}
 \put(7.1,8.2){\footnotesize{C}}
 \put(7.1,1.5){\footnotesize{D}}

 \put(4.7,6.8){\footnotesize{A$'$}}
 \put(4.7,3.1){\footnotesize{B$'$}}
 \put(4.65,8.2){\footnotesize{C$'$}}
 \put(4.65,1.5){\footnotesize{D$'$}}

 \end{picture}
 \caption{Phase portrait of \eqref{eq:pqodes} with $k=-1$ (cf. figure \ref{fig:phaseplane_k=-1}); orbits correspond to $H(p,q,;-1)=E_-$. Closed orbits (blue) represent energies $0<E_-<1$; unbounded orbits (gray) represent energies $E_->1$; the black orbit represents $E_-=1$; the dashed red orbit represents $E_-=E^*\approx 1.462$. Arrows indicate the direction of travel in each quadrant as $y$ increases. Points A--D$'$ are discussed in the text.}
 \label{fig:phaseplane_k=-1_imperfect}

\end{center}
\end{figure}

Each solution is characterised by an energy $E_-$ and by the starting value of $p=p_{\beta}=\tan(\beta)$. It is readily seen from the phase plane (figure \ref{fig:phaseplane_k=-1_imperfect}) that for any value of $E_-$ there are two possible solutions; it is convenient to consider $E_-\leq1$ and $E_->1$ separately.

For $E_-\leq 1$ solutions either run from A: $(p_{\beta},q_{\beta})$ to A$'$: $(-p_{\beta},q_{\beta})$, or from B: $(p_{\beta},-q_{\beta})$ to B$'$: $(-p_{\beta},-q_{\beta})$, where
\begin{equation}\label{eq:qbeta}
 q_{\beta} = \left(\dfrac{2(1+(E_--1)\sqrt{1+p_{\beta}^2})}{\sqrt{1+p_{\beta}^2}}\right)^{1/2}.
\end{equation}
Examples of these points are plotted on figure \ref{fig:phaseplane_k=-1_imperfect}, taking $E_-=0.6$ for AA$'$ and BB$'$, and $E_-=1.2$ for CC$'$ and DD$'$.

\begin{remark}
Solutions cannot run, for example, from A to B$'$, because such solutions cannot satisfy the boundary conditions $\zeta(\pm a) = 0$.
\end{remark}

Note that for a given value of $p_{\beta}>0$ there is a minimum attainable value of the energy, given by
\begin{equation}\label{eq:Emin}
 E_- = E_{\min}(p_{\beta}) = 1-\dfrac{1}{\sqrt{1+p_{\beta}^2}}.
\end{equation}
When $E_-=E_{\min}(p_{\beta})$, the points A and B (and likewise A$'$ and B$'$) are identical.

Solutions of the form AA$'$, BB$'$ and CC$'$ all correspond to trajectories that lie entirely in the $k=-1$ phase plane. To construct solutions of the form DD$'$ we must connect trajectories in the $k=-1$ plane with trajectories in the $k=+1$ plane as in \S\ref{subsec:largeenergy}. (We omit the details here for brevity, as they contain no new ideas.)

Figure \ref{fig:pb=1_profiles} illustrates the profiles for $E_-=0.6$ (AA$'$ and BB$'$) and $E_-=1.2$ (CC$'$ and DD$'$). Note that the small-area solution in each case can be obtained by truncating the high-area solution at an appropriate horizontal level, and indeed each solution can be obtained by truncating the perfectly wetting solution for the same energy. (This construction is made more explicit in \cite{PerazzoGratton2004}.)

\begin{figure}[tbp]
\begin{center}
\setlength{\unitlength}{1.0cm}

 \begin{picture}(14,7.5)

 \put(0.0,0.0){\epsfig{figure={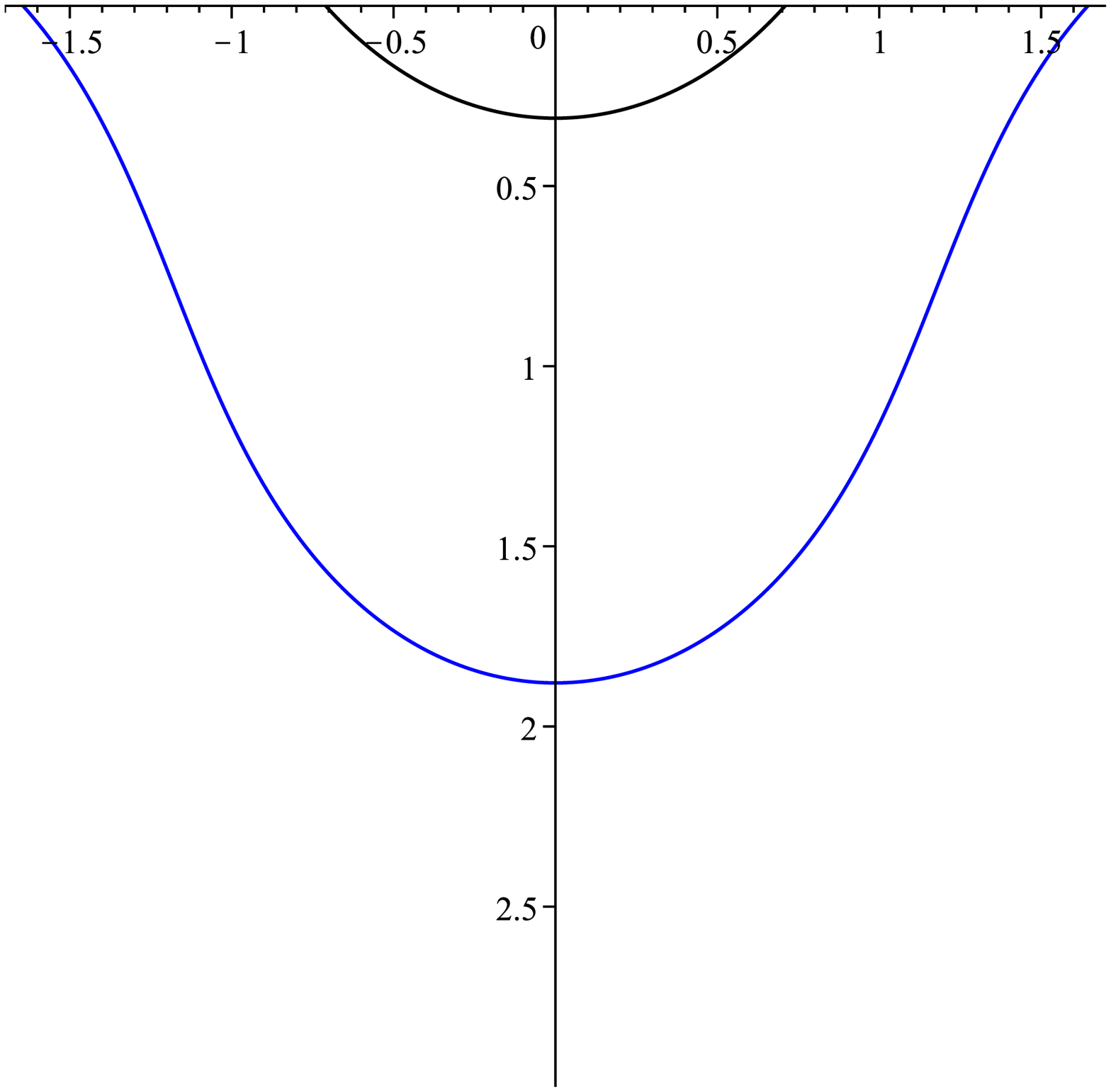},width=7\unitlength}}
 \put(7.0,0.0){\epsfig{figure={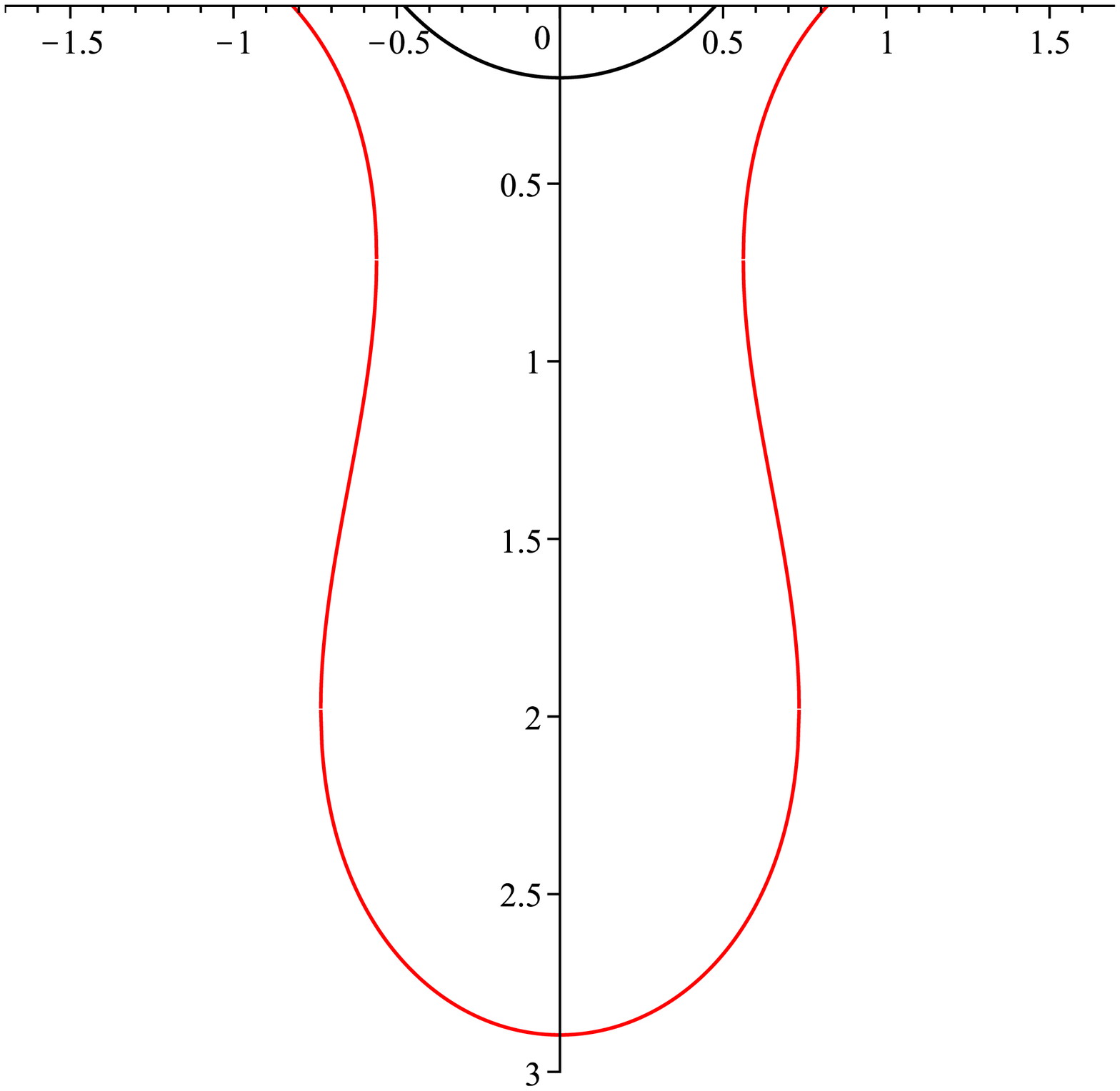},width=7\unitlength}}

 \put(0.0,7.0){\footnotesize{(a)}}
 \put(7.0,7.0){\footnotesize{(b)}}

 \put(4.45,6.05){\footnotesize{\textcolor{black}{AA$'$}}}
 \put(5.0,4.3){\footnotesize{\textcolor{blue}{BB$'$}}}
 \put(10.9,6.1){\footnotesize{\textcolor{black}{CC$'$}}}
 \put(11.9,4.3){\footnotesize{\textcolor{red}{DD$'$}}}
 
 \put(3.75,3.3){\scriptsize{$z$}}
 \put(5.5,7.0){\scriptsize{$y$}}
 \put(10.75,3.3){\scriptsize{$z$}}
 \put(12.5,7.0){\scriptsize{$y$}}
 
 \end{picture}
 \caption{The rivulet profiles $z=\zeta(y)$ for (a) $E_-=0.6$ and (b) $E_-=1.2$. In (a), the profile in black corresponds to the AA$'$ trajectory and the profile in blue corresponds to the BB$'$ trajectory. In (b), the profile in black corresponds to the CC$'$ trajectory and the profile in red corresponds to the DD$'$ trajectory. (See figure \ref{fig:phaseplane_k=-1_imperfect} in each case.)}
 \label{fig:pb=1_profiles}

\end{center}
\end{figure}

Because the shape of the solutions depends only on $E_-$ and not on $p_{\beta}$, pinch-off occurs at the same energy $E_-=E^*$ as for the perfectly wetting solutions. However, pinch-off does not affect CC$'$ solutions as these do not have an overhang.

\subsubsection{Bifurcation structure\label{subsec:imperfect-bifurcation}}

Figures \ref{fig:areavswidth_pb=0_1} and \ref{fig:areavswidth_pb=1} show the bifurcation diagram in the $(a,A)$-plane for $p_{\beta}=0.1$ and $p_{\beta}=1$ respectively.

\begin{figure}[tbp]
\begin{center}
\setlength{\unitlength}{1.0cm}

 \begin{picture}(12,10.5)

 \put(1.0,0.0){\epsfig{figure={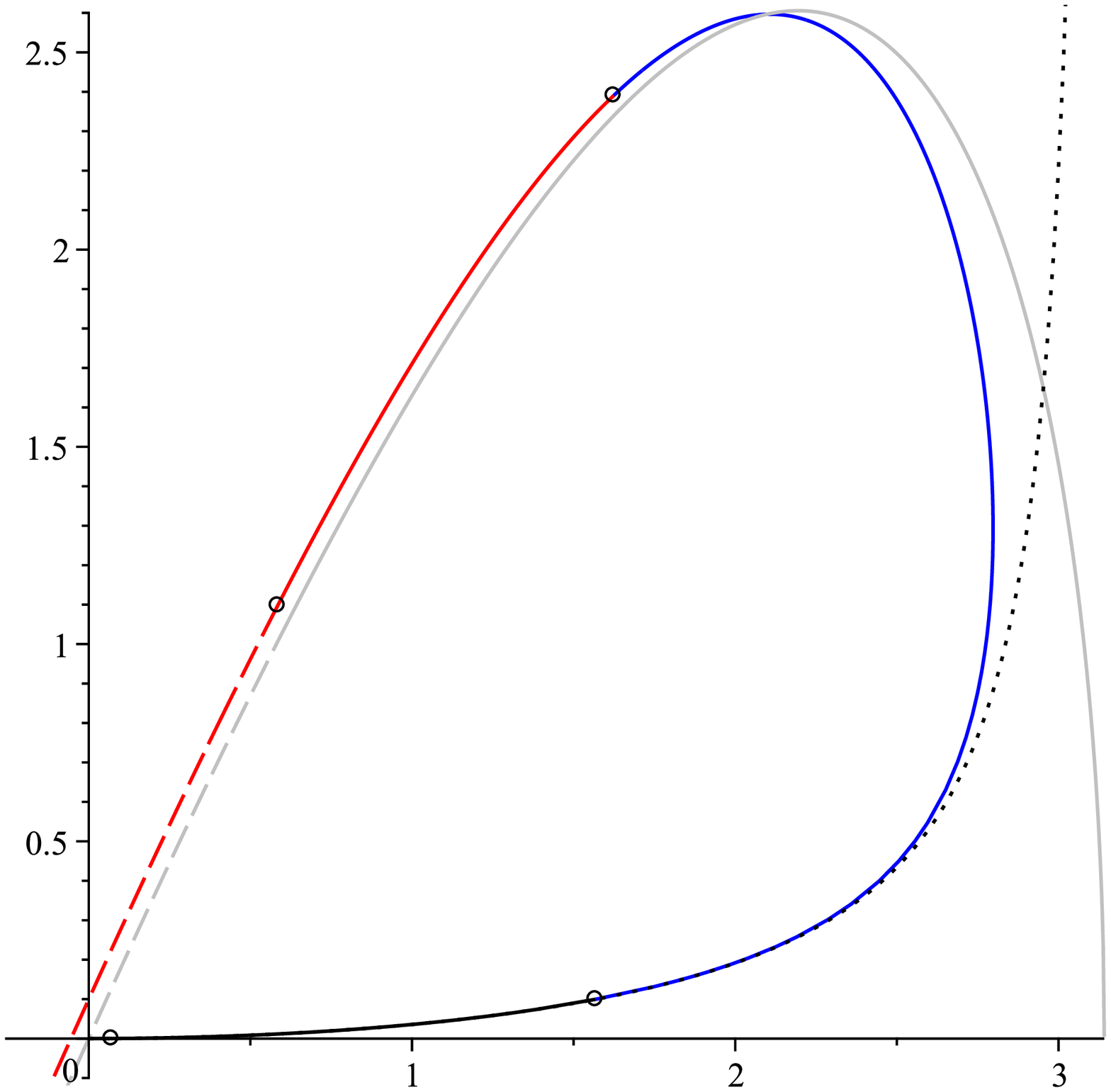},width=10\unitlength}}

 \put(6.1,0.2){\footnotesize{$a$}}
 \put(1.0,5.3){\footnotesize{$A$}}

 \put(4.9,1.45){\footnotesize{$E_-=E_{\min}$}}
 \put(6.75,8.85){\footnotesize{$E_-=1$}}
 \put(3.8,4.4){\footnotesize{$E_-=E^*$}}
 \put(2.2,1.05){\footnotesize{$E_-=1$}}

 \put(4.3,1.1){\scriptsize{AA$'$}}
 \put(9.1,6.1){\color{blue}{\scriptsize{BB$'$}}}
 \put(4.2,6.8){\color{red}{\scriptsize{DD$'$}}}
 
 \end{picture}
 \caption{The half-area $A$ of a rivulet plotted against the half-width $a$, for $p_{\beta}=0.1$. The black portion of the curve represents solutions of the form AA$'$ and CC$'$; the blue portion represents solutions of the form BB$'$; the red portion represents sections of the form DD$'$. The gray curve is for $p_{\beta}=0$ (cf. figure \ref{fig:areavswidth}). Dashed portions of each curve represent solutions beyond pinch-off. The dotted line is the asymptotic result \eqref{eq:Aaasympt} from lubrication theory.}
 \label{fig:areavswidth_pb=0_1}

\end{center}
\end{figure}

For $p_{\beta}>0$, $q=0$ is no longer a solution to the boundary-value problem and the pitchfork bifurcation (Proposition \ref{prop:pitchfork}) is replaced by a saddle-node bifurcation (figure \ref{fig:areavswidth_pb=0_1}). There are two solution branches, so for a given half-width $a$ there are two solutions: a small-$A$ solution and a large-$A$ solution.

\begin{remark}
Lubrication theory is premised on $|p|\ll 1$ (see \S\ref{sec:lubrication}), i.e. lubrication solutions are approximations to AA$'$ and CC$'$ solutions in the vicinity of the positive $q$-axis. Hence lubrication theory describes only the small-$A$ solution branch. 

From the solution \eqref{eq:zetasol_lub_pb} in the lubrication limit we obtain the asymptotic result
\begin{equation}\label{eq:Aaasympt}
 A(a) = p_{\beta}\left(1 - a\cot(a)\right),
\end{equation}
which is plotted as a dotted line in figure \ref{fig:areavswidth_pb=0_1}. The asymptotic approximation \eqref{eq:Aaasympt} is usually described (cf. \cite{WilsonDuffy2005}) as valid for $p_{\beta}\ll 1$; however, as figure \ref{fig:areavswidth_pb=0_1} illustrates, it is not uniformly valid in $a$ because of the change in the bifurcation structure when $p_{\beta}> 0$.
\end{remark}

\begin{figure}[tbp]
\begin{center}
\setlength{\unitlength}{1.0cm}

 \begin{picture}(12,10.5)

 \put(1.0,0.0){\epsfig{figure={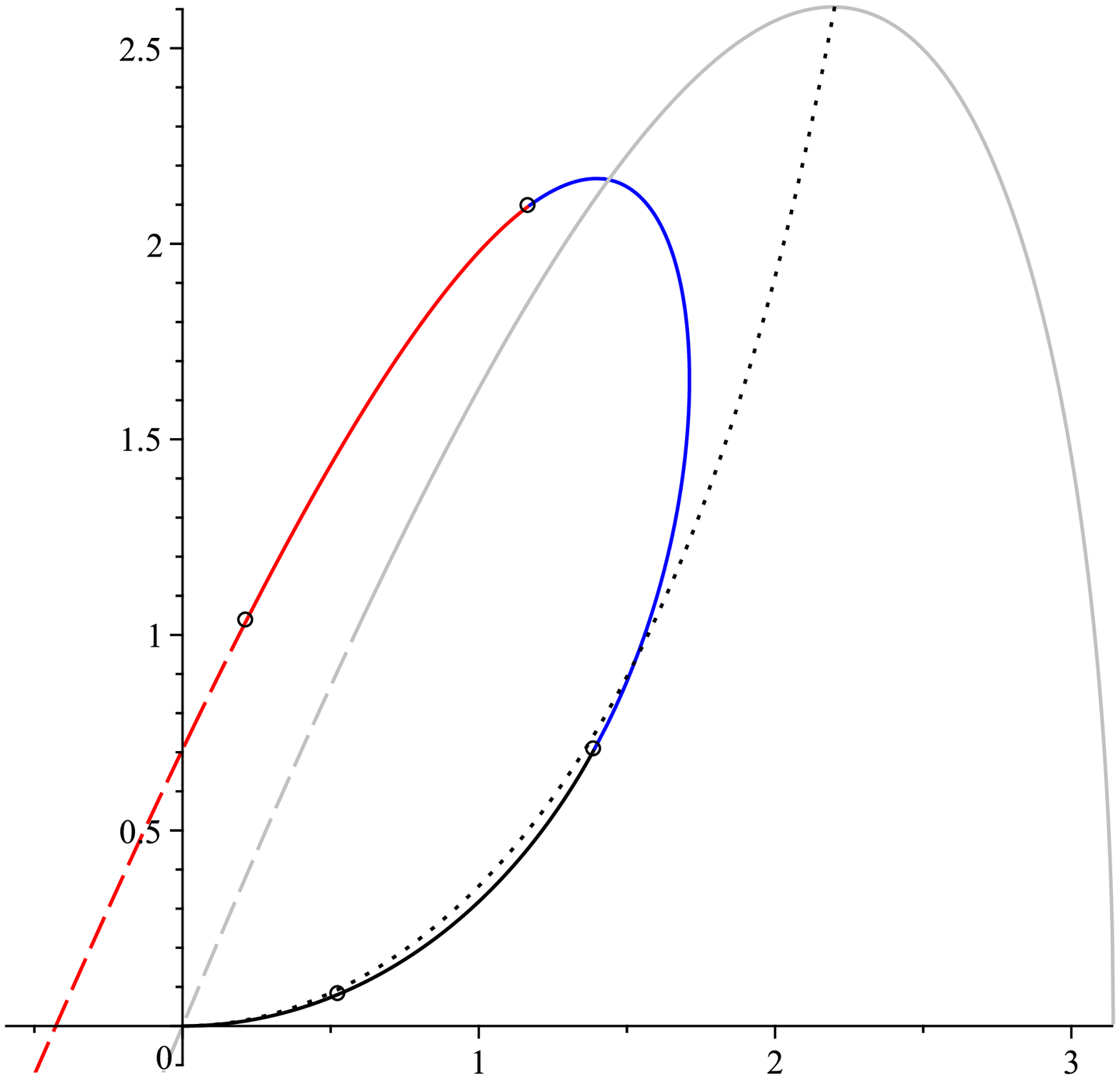},width=10\unitlength}}

 \put(6.1,0.2){\footnotesize{$a$}}
 \put(1.0,5.3){\footnotesize{$A$}}

 \put(6.4,3.1){\footnotesize{$E_-=E_{\min}$}}
 \put(4.8,8.25){\footnotesize{$E_-=1$}}
 \put(3.4,4.3){\footnotesize{$E_-=E^*$}}
 \put(3.0,1.3){\footnotesize{$E_-=1$}}

 \end{picture}
 \caption{The half-area $A$ of a rivulet plotted against the half-width $a$, for $p_{\beta}=1$. The black portion of the curve represents solutions of the form AA$'$ and CC$'$; the blue portion represents solutions of the form BB$'$; the red portion represents sections of the form DD$'$. The gray curve is for $p_{\beta}=0$ (cf. figure \ref{fig:areavswidth}). Dashed portions of each curve represent solutions beyond pinch-off. The dotted line is the asymptotic result \eqref{eq:Aaasympt} from lubrication theory.}
 \label{fig:areavswidth_pb=1}

\end{center}
\end{figure}

As $p_{\beta}$ increases, the saddle-node bifurcation occurs at smaller values of $a$ and, as might be expected, the lubrication approximation becomes less accurate (figure \ref{fig:areavswidth_pb=1}).

Figure \ref{fig:bifout} illustrates further how the bifurcation diagram in the $(a,A)$-plane changes as $p_{\beta}$ increases.

\begin{figure}[tbp]
\begin{center}
\setlength{\unitlength}{1.0cm}

 \begin{picture}(12,10.5)

 \put(2.0,0.0){\epsfig{figure={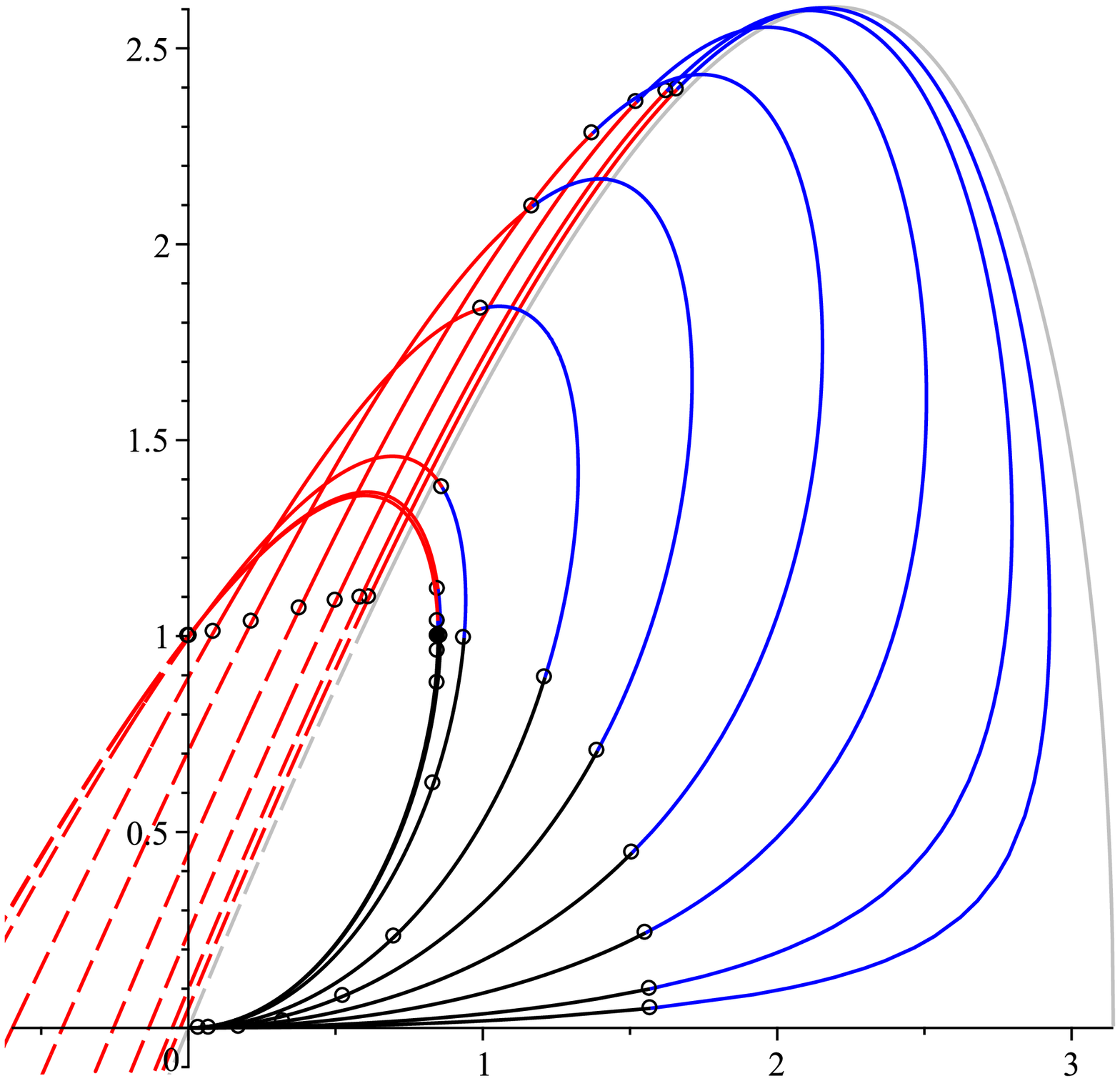},width=10\unitlength}}
 \put(0.0,7.0){\epsfig{figure={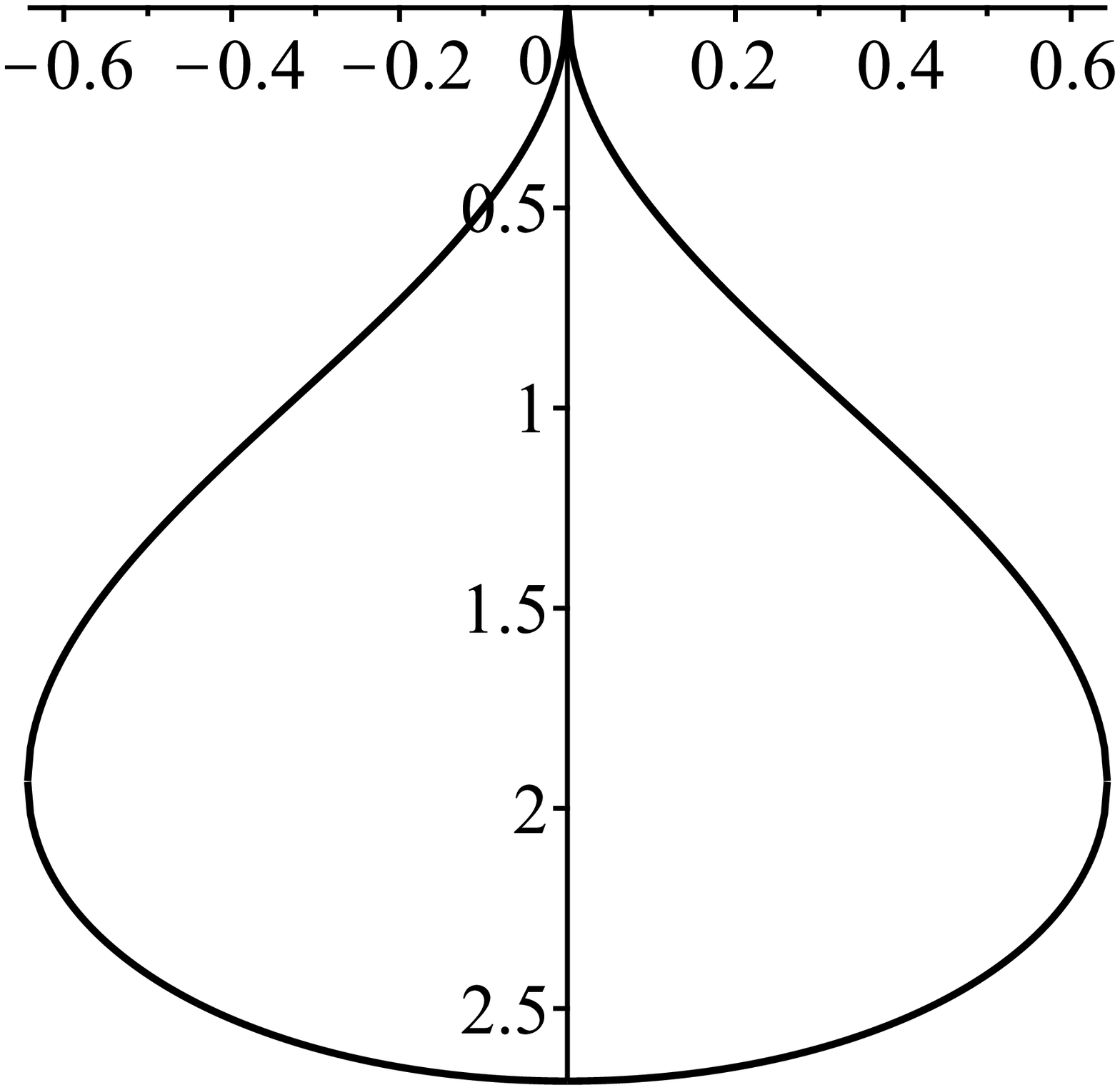},width=3\unitlength}}
 \put(0.0,3.0){\epsfig{figure={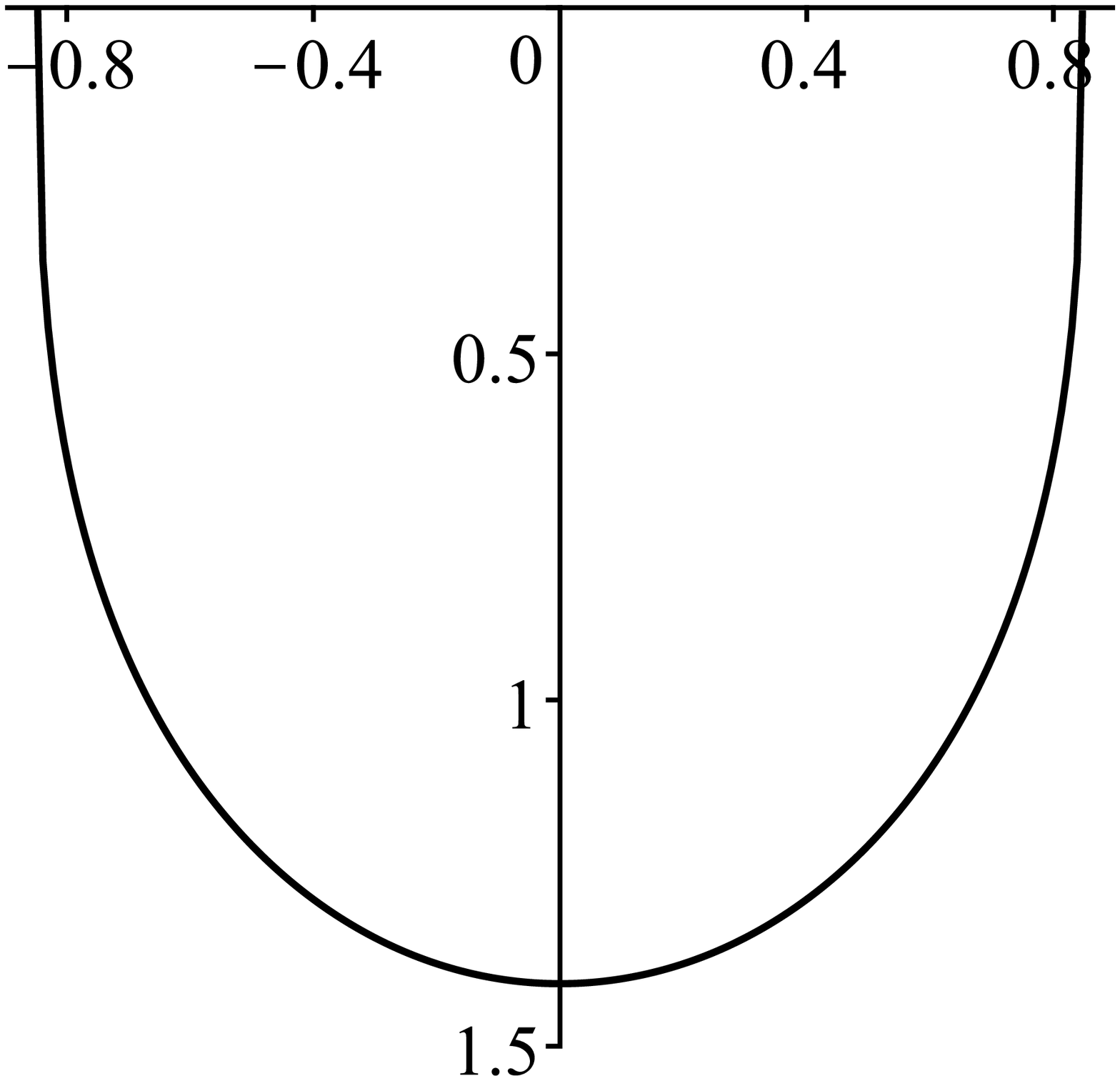},width=3\unitlength}}

 \put(7.1,0.2){\footnotesize{$a$}}
 \put(3.0,6.5){\footnotesize{$A$}}

 \end{picture}
 \caption{The half-area $A$ of a rivulet plotted against the half-width $a$, for $p_{\beta}=0.05$, $0.1$, $0.25$, $0.5$, $1$, $2$, $10$, $100$ and $1000$. The black portion of each curve represents solutions of the form AA$'$ and CC$'$; the blue portion represents solutions of the form BB$'$; the red portion represents sections of the form DD$'$. The gray curve is for $p_{\beta}=0$ (cf. figure \ref{fig:areavswidth}). Dashed portions of each curve represent solutions beyond pinch-off. Open circles mark the points $E_-=1$, $E_-=E_{\min}(p_{\beta})$ and $E_-=E^*$ on each curve; the solid circle marks $(a_{\infty},A_{\infty})$. The upper inset shows the profile of the rivulet in the limit $E_-=E^*$, $p_{\beta}\to\infty$; the lower inset shows the profile of the rivulet in the limit $E_-=1$, $p_{\beta}\to\infty$.}
 \label{fig:bifout}

\end{center}
\end{figure}

The limiting `pinch-off' solution ($E_-=E^*$) in the limit as $p_{\beta}\to\infty$ is shown in the upper inset in figure \ref{fig:bifout}. Pinch-off occurs for $\zeta_{\mathrm{p}}=0$, giving a contact angle of $\pi/2$; this solution therefore has $a=0$ by construction, and by Proposition \ref{prop:pinchoffarea} its half-area is $A = 1$.

In the limit as $p_{\beta}\to\infty$, $E_{\min}(p_{\beta})\to 1$, and so the section of the curve corresponding to BB$'$ solutions disappears. The limiting solution without overhang (the lower inset in figure \ref{fig:bifout}) corresponds to $E_-=1$ and $p_{\beta}\to \infty$; its half-width is thus given by
\begin{equation}\label{eq:alimitAA}
 a_{\infty} = \int_0^{\sqrt{2}}\dfrac{q^2}{[(2-q^2)(2+q^2)]^{1/2}}\d q = 2E\left(\frac{1}{\sqrt{2}}\right)-K\left(\frac{1}{\sqrt{2}}\right) \approx 0.847,
\end{equation}
where $E$ and $K$ are elliptic functions as before. The half-area of this limiting solution is given, using \eqref{eq:dqdy}, by
\begin{equation}\label{eq:AlimitAA}
 A_{\infty} = \int_0^{\sqrt{2}}q\dfrac{\d y}{\d q}\d q = \int_0^{\sqrt{2}}\dfrac{q}{\left[4q^{-4}-1\right]^{1/2}}\d q = 1.
\end{equation}

\subsection{Hydrophobic substrate ($\beta>\pi/2$)\label{sec:hydrophobic}}

When $\beta>\pi/2$, each solution corresponds to a trajectory that starts and ends in the $k=+1$ plane and is connected via the $k=-1$ plane. (In fact, every such solution corresponds to part of a DD$'$ solution as described above.) We omit the details of the construction as they can readily be deduced from \S\ref{subsec:largeenergy}, and give only a brief account of the results.

For given admissible values of $E_-$ and $p_{\beta}<0$, we define
\begin{equation}
 q_{\beta} = \sqrt{2\left(E_--1-\frac{1}{(1+p_{\beta}^2)^{1/2}}\right)}.
\end{equation}
Solutions either run from E to E$'$ or from F to F$'$, where E: $(p_{\beta},q_{\beta})$, E$'$: $(-p_{\beta},q_{\beta})$,
F: $(p_{\beta},-q_{\beta})$, F$'$: $(-p_{\beta},-q_{\beta})$, all in the $k=+1$ phase plane.

For a given value of $p_{\beta}$, the energy $E_-$ of an EE$'$ or FF$'$ solution must satisfy
\begin{equation}
 E_- > E_{\beta} = 1 + \dfrac{1}{(1+p_{\beta}^2)^{1/2}}.
\end{equation}

\begin{figure}[tbp]
\begin{center}
\setlength{\unitlength}{1.0cm}

 \begin{picture}(12,10.5)

 \put(1.0,0.0){\epsfig{figure={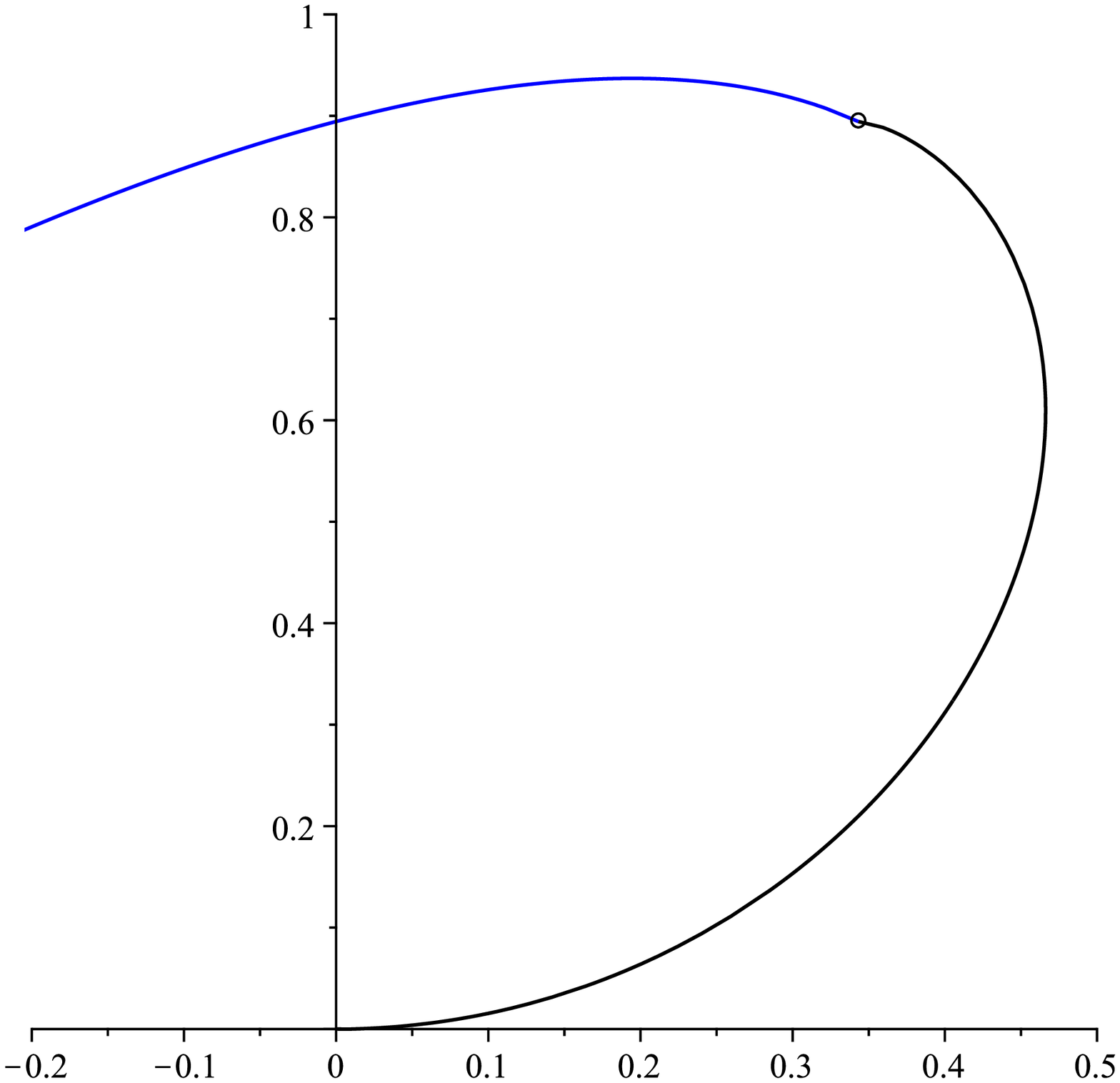},width=10\unitlength}}

 \put(6.65,0.2){\footnotesize{$a$}}
 \put(3.8,5.05){\footnotesize{$A$}}

 \put(8.65,8.9){\footnotesize{$E_-=E_{\beta}$}}
 \put(4.2,1.1){\footnotesize{$E_-\to\infty$}}

 \put(10.5,3.5){\epsfig{figure={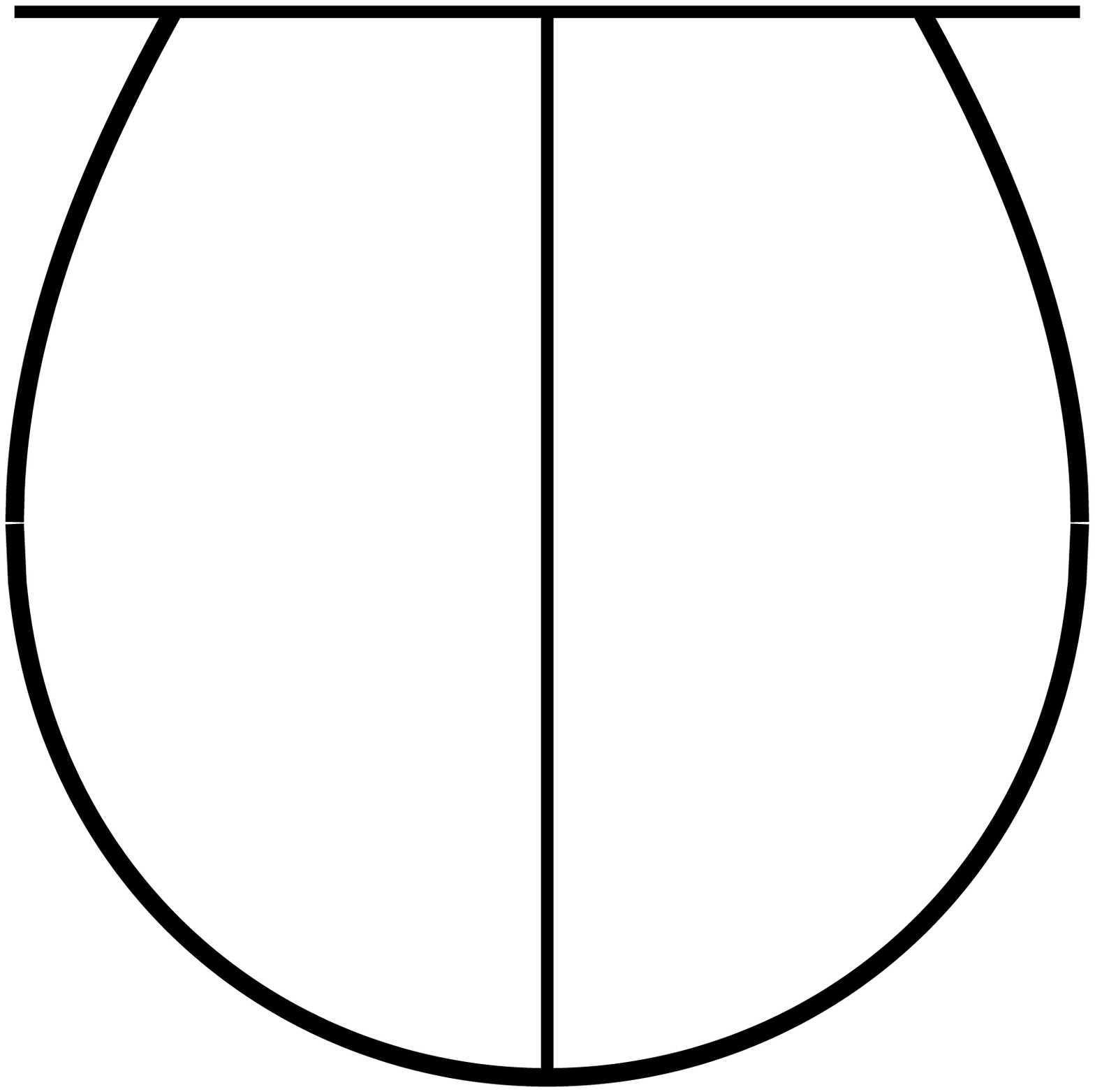},width=2\unitlength}}
 \put(5.3,6.8){\epsfig{figure={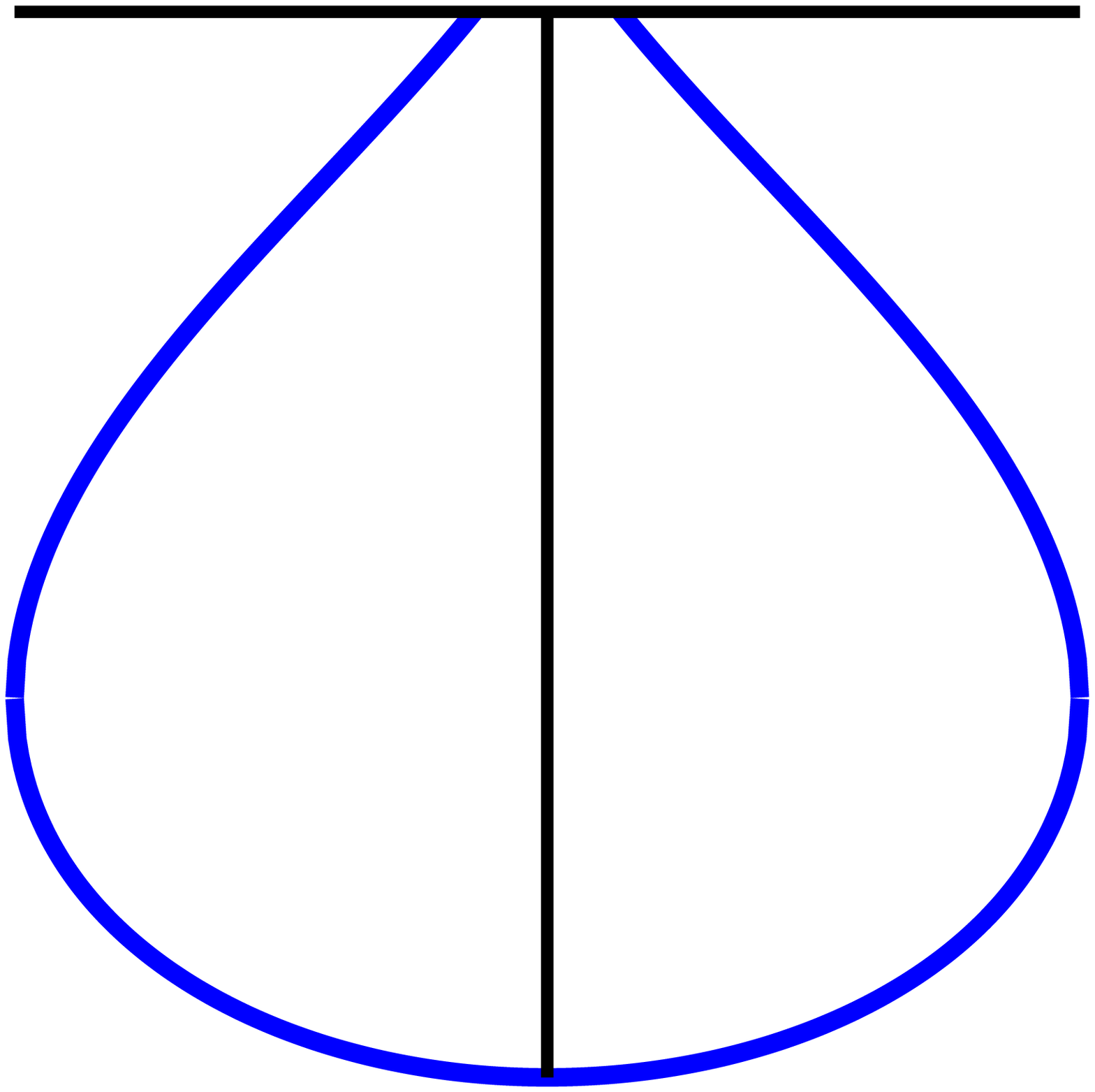},width=2\unitlength}}
 \put(1.3,5.7){\epsfig{figure={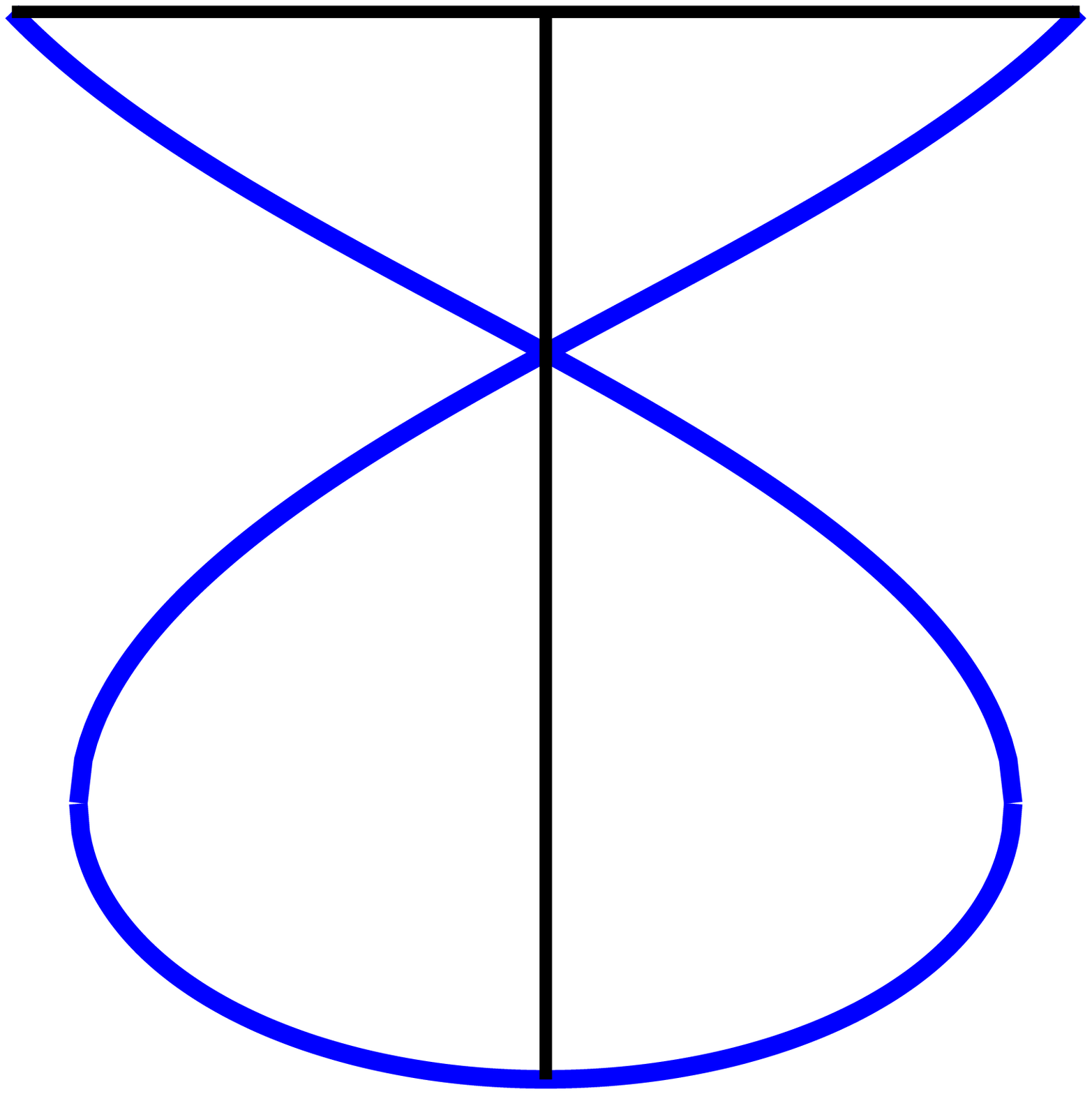},width=2\unitlength}}

 \put(8.5,2.1){\scriptsize{EE$'$}}
 \put(3.0,8.7){\color{blue}{\scriptsize{FF$'$}}}

 \end{picture}
 \caption{The half-area $A$ of a rivulet plotted against the half-width $a$, for $p_{\beta}=-2$. The black portion of the curve represents solutions of the form EE$'$; the blue portion represents solutions of the form FF$'$. The insets (not to scale) indicate the form of the rivulet in each case.}
 \label{fig:areavswidth_pb=-2}

\end{center}
\end{figure}

Figure \ref{fig:areavswidth_pb=-2} shows the bifurcation diagram in the $(a,A)$-plane, and typical rivulet profiles, for a hydrophobic rivulet with $p_{\beta}=-2$. The EE$'$ branch (black) consists of physically valid solutions for all $E_-\geq E_{\beta}$; as $E_-\to\infty$, $(a,A)\to (0,0)$. The FF$'$ branch (blue) consists of physically valid solutions for $E_{\beta} \leq E_- \lesssim 1.525$, at which point the half-width $a=0$; for values of $E_- \gtrsim 1.525$, the solutions self-intersect (cf. figure 3 of \cite{Sokurov2020}). Note that the value of $E_-$ beyond which the solutions self-intersect depends on $p_{\beta}$, but must always be greater than $E^*$.

\begin{figure}[tbp]
\begin{center}
\setlength{\unitlength}{1.0cm}

 \begin{picture}(12,10.5)

 \put(2.0,0.0){\epsfig{figure={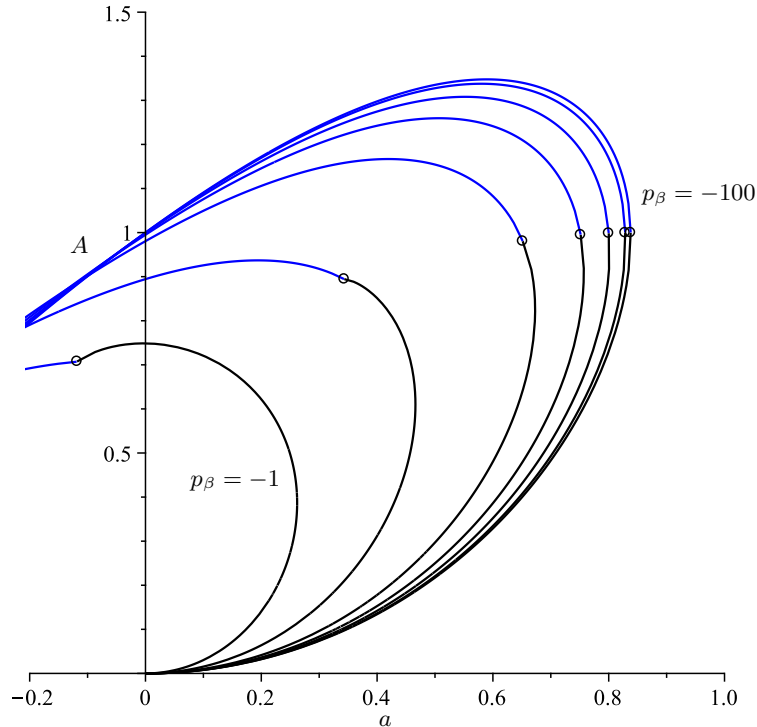},width=10\unitlength}}

 \put(7.1,0.2){\footnotesize{$a$}}
 \put(3.0,6.5){\footnotesize{$A$}}

 \put(4.6,3.4){\footnotesize{$p_{\beta}=-1$}}
 \put(10.6,7.2){\footnotesize{$p_{\beta}=-100$}}

 \end{picture}
 \caption{The half-area $A$ of a rivulet plotted against the half-width $a$, for $p_{\beta}=-1$, $-2$, $-5$, $-10$, $-20$, $-50$, and $-100$. The black portion of each curve represents solutions of the form EE$'$; the blue portion represents solutions of the form FF$'$. Open circles mark the point $E_-=E_{\beta}$ on each curve.}
 \label{fig:bifout_hydrophobic}

\end{center}
\end{figure}

Figure \ref{fig:bifout_hydrophobic} illustrates how the bifurcation diagram in the $(a,A)$-plane changes as $p_{\beta}$ changes. As $p_{\beta}\to-\infty$, the contact angle $\beta \to \pi/2$ and so the curves approach the same limit as those for the hydrophilic problem in the limit $p_{\beta}\to\infty$ (see figure \ref{fig:bifout}). As $|p_{\beta}|$ becomes smaller, the range of physically valid solutions decreases; for sufficiently small $|p_{\beta}|$, only solutions of the form EE$'$ are valid.

\section{Discussion\label{sec:discussion}}

We have presented an approach, based on phase-plane methods, to the problem of describing the free surface of a pendent rivulet. Although this problem has previously been tackled by other methods (e.g. \cite{Tanasijczuketal2010}), the present approach allows us to make greater analytical progress, and in particular to carry out a more thorough bifurcation analysis.

In particular, our approach elucidates why lubrication theory fails to capture the behaviour of perfectly wetting pendent rivulets: the lubrication approach corresponds to linearising around a centre in the phase plane. This is analogous to linearising a nonlinear oscillator about a non-hyperbolic equilibrium; another classical example occurs when linearising the Euler elastica around a bifurcation point \cite{Brown2014}. For imperfectly wetting pendent rivulets, lubrication theory provides a good approximation to one of the two solution branches when the rivulet is narrow, but the asymptotic approximation even to this branch is not uniform in the half-width $a$.

A feature of our approach which appears to be novel is that we `splice together' solutions from orbits in two related phase planes. We suggest that this could be a useful technique in other problems that can be formulated such that the governing equation changes discontinuously but the solution itself does not.

Finally, we note that our approach can readily be applied to the corresponding problem for sessile rivulets. For sessile rivulets on hydrophilic surfaces, $0<\beta<\pi/2$, solutions correspond to orbits in the $k=+1$ phase plane. Because the origin in this phase plane is a saddle point, linearising about it correctly captures the behaviour \cite{PerazzoGratton2004}: in particular, lubrication theory correctly predicts \cite{WilsonDuffy2005} that there are no solutions for perfectly wetting sessile rivulets. For sessile rivulets on hydrophobic surfaces, $\beta>\pi/2$, solutions correspond to orbits that start and end in the $k=-1$ phase plane and are connected via the $k=+1$ phase plane; they may be constructed analogously to those described in \S\ref{sec:hydrophobic}.

Finally, we note that the approach here has the potential to be applied to the related problem of horizontal or vertical liquid bridges (e.g. \cite{Haynes2016}) without assuming the smallness of any parameters, and we recommend this as a direction for future work.

\null

\textbf{Acknowledgements.} We are grateful to Brian R. Duffy and Stephen K. Wilson for discussions of the rivulet problem over the course of several years.


\end{document}